\def\JAS{1} 
\def\JASA{0}

\documentclass[]{interact} 

\if\JASA1
{
\addtolength{\oddsidemargin}{-.5in}%
\addtolength{\evensidemargin}{-.5in}%
\addtolength{\textwidth}{1in}%
\addtolength{\textheight}{-.3in}%
\addtolength{\topmargin}{-.8in}%
} \fi

\usepackage{mathtools,breqn,comment,graphicx}
\usepackage{xcolor,epstopdf}
\usepackage[caption=false]{subfig}

\usepackage[numbers,sort&compress]{natbib}
\bibpunct[, ]{[}{]}{,}{n}{,}{,}

\theoremstyle{plain}
\newtheorem{theorem}{Theorem}[section]
\newtheorem{lemma}[theorem]{Lemma}
\newtheorem{corollary}[theorem]{Corollary}

\newcommand\blue{\color{black}}

\theoremstyle{definition}
\newtheorem{definition}[theorem]{Definition}
\newtheorem{example}[theorem]{Example}

\theoremstyle{remark}
\newtheorem{remark}{Remark}

\newtheorem{property}{Property}[section]

\newcommand\apj{\emph{The Astrophysical Journal}}%

\newcommand\mnras{\emph{Monthly Notices of the Royal Astronomical Society}}%

\newcommand\aap{\emph{Astronomy and Astrophysics}}

\usepackage{amsmath,amsthm}
\usepackage{graphicx,psfrag,epsf}
\usepackage{enumerate}
\usepackage{natbib}
\usepackage{url} 
\usepackage{amssymb,amsmath}
\usepackage{color,xcolor,nicefrac}

\def\blue{\color{black}}
\def\black{\color{black}}

\def\cstat{${C}$ statistic}
\def\cmin{${C}_{\mathrm{min}}$}

\def\chandra{\emph{Chandra}}

\newcommand{\blind}{0} 
\begin{document}
\if\JAS1
{
\articletype{Original Research Article}
}\fi
\if\JASA1
{
\articletype{Theory and methods}
\def\spacingset#1{\renewcommand{\baselinestretch}%
{#1}\small\normalsize} \spacingset{1}
\spacingset{1.5}
}\fi

\title{A semi--analytical solution to the maximum--likelihood
fit of Poisson data to a linear  model using the Cash statistic}
\if0\blind{
\author{
\name{Massimiliano Bonamente\textsuperscript{a}\thanks{
The author gratefully acknowledges support of \textit{NASA} 
\chandra\ grant AR6-17018X.} and
David Spence\textsuperscript{a}}
\affil{\textsuperscript{a} Department of Physics and Astronomy,
University of Alabama in Huntsville, Huntsville AL 35899 (U.S.A)}}
}\fi
\maketitle

\begin{abstract}

The \emph{Cash} statistic, also known as the \cstat,  
is commonly used 
for the analysis of low--count Poisson data, including data
with null counts for certain values of the independent variable. 
The use of this statistic is especially attractive for
low--count data that cannot be combined, or re--binned, 
without loss of resolution. 
	{\blue This paper presents a new maximum--likelihood solution for the
best--fit parameters of a linear model using the
	Poisson--based \emph{Cash} statistic.}

The solution presented in this paper provides a 
new and simple method to measure the best--fit parameters
of a linear model for any Poisson--based data, including data with
null counts. In particular, the method 
enforces the requirement that the best--fit linear model be non--negative throughout
the support of the independent variable.
The method is summarized in a simple algorithm to
fit Poisson counting data
of any size and counting rate with a linear model,
by--passing entirely the use of the traditional $\chi^2$ statistic.
\end{abstract}

\begin{keywords}
Probability; Statistics; Maximum--likelihood methods; Cash statistic; Parameter estimation
\end{keywords}

\section{Introduction}

The maximum--likelihood modelling of integer--valued Poisson data
can be accomplished with the use of the \emph{Cash}, or \cstat,
first proposed by \cite{cash1979}.
{\blue The Cash statistic
applies to a variety of counting data in use across the sciences.
One example is the counting of photons as a function of energy or wavelength, 
as commonly done by photon--counting detectors used in
astronomy \citep[e.g.,][]{bonamente2019}. Another example is the number or percentage 
of votes for a candidate in different
precincts or polling stations. In counting experiments such as these, the collected
data are in the form of  independent integer--valued variables.
The behavior of these variables as a function
of an independent variable (such as photon energy or number of voters in a precinct) 
 can be modelled
with the aid of the Cash statistic, which is obtained from
the  logarithm of the likelihood of the data with a model
for the distribution of the counts.}

It is well established that the asymptotic distribution of the \cstat,
in the large--count limit, is a $\chi^2$ distribution 
\citep[e.g.][]{cash1979,kaastra2017}. 
This limit is a result of the asymptotic convergence of a Poisson
distribution with a Gaussian distribution of same mean and variance,
which occurs for large values of the Poisson mean.
It is straightforward to study the distribution 
of the \cstat\ as a function of the parent mean $\mu$,
when the parent mean is specified \emph{a priori}, e.g., when the fitting model
has no free parameters. Such calculations are reported in 
\cite{cash1976,kaastra2017,bonamente2020},
showing, among other results, that the expectation of $C$ is significantly lower
than the expectation of a $\chi^2$ distribution with the same number of degrees of freedom.

The use of the \cstat\ for integer--valued, counting Poisson data is to be preferred
to the use of the more common $\chi^2$ distribution.
First, even in the large--count limit, use of the  $\chi^2$ fit statistic
leads to a bias in the best--fit parameters, due to the approximation
of the Gaussian variance with the measured Poisson counts \citep[e.g., ][]{humphrey2009}.
Second, use of the $\chi^2$ statistic often requires the combination
of datapoints, often referred to as \emph{binning} of the independent variable, 
to reach  a sufficient number of counts in each independent data point.
Such binning may result in an undesirable reduction in the resolution of the data, especially in the presence of
sharp features in narrow intervals of the independent variable, 
such as emission or absorption lines.
The \cstat, on the other hand, can be used on unbinned data that make use of the
full resolution of the data.

Use of the \cstat\ also comes with a number of challenges. 
First, it is not known exactly how free model parameters affect the distribution
of \cmin -- i.e., the \cstat\ 
obtained when optimizing variable model parameters -- especially in the low--count regime.
A study of the distribution of \cmin\ for a simple one--parameter constant model
was reported by \cite{bonamente2020}, although those results do not 
directly apply to more complex models such as the linear model.

Another challenge is the 
numerical complexity of the Poisson distribution and the associated
\cstat, which limits
the ability to obtain analytical solutions for the best--fit parameters
via the maximum--likelihood criterion. A key illustration of this challenge is that 
the simple linear model, which has an analytical solution for 
its best--fit model parameters and
their covariance matrix when using the $\chi^2$ statistic \citep[e.g.][]{bevington2003,bonamente2017book},
does not have an equally simple solution when using the \cstat.

This paper addresses the latter problem 
by presenting a new semi--analytical method to identify 
the maximum--likelihood solution of the parameters of a linear model
using the \cstat.
The method consists of the numerical solution of a simple analytical equation
that determines the best--fit value of one of the two parameters, 
and the use of an analytical function to calculate the other parameter.
It is also shown that not all Poisson data sets can be fit to an unconstrained
linear model, since the resulting best--fit model
may become negative, and therefore not usable for calculation of the Poisson--based \cstat.
In those cases, a simple generalization of the linear model is proposed
that enforces the non--negative requirement for the best--fit model. Such 
a  generalization
ensures that data of all sizes and counting rates can be fit with a 
linear model, and that such model is unique.
The results presented in this paper therefore ease the challenges presented 
by the numerical complexity of the Poisson distribution, by providing a simple
semi--analytical method to find the best--fit parameters of the linear model,
and making it possible to study the distribution of $C_{min}$ in the
low--count regime.

This paper is structured as follows. Section~\ref{sec:methods} introduces the maximum--likelihood method
and the Poisson--based \cstat, Section~\ref{sec:ML} presents the equations
for the model parameters and Section~\ref{sec:analytical} the solution of those equations.
Section~\ref{sec:acceptability} discusses conditions for the non--negativity of the best--fit linear model
and Section~\ref{sec:generalized} presents the {\blue extended} linear model that ensures an acceptable
model for all datasets. Finally, Section~\ref{sec:discussion} contains a discussion and conclusions.

\section{Methods for the maximum--likelihood analysis of Poisson data}
\label{sec:methods}
\subsection{Data model and the \cstat}
The data model considered in this paper is $N$ independent 
integer--valued measurements $y_i$, each Poisson--distributed
and measured at a fixed value $x_i$ of the independent variable $x$.
The data can also be viewed as originating from $M=\sum y_i$ independent events
that are sorted into $N$ independent "bins", each of size $\Delta x_i$ and centered
at $x_i$ with $y_i$ counts. This is a common  type of data for 
the physical sciences; for example,
the independent variable $x$ may be the wavelength of collected photons,
and $y$ the number of photons collected in a given wavelength range,
binned according to the resolution of the instrument. 
The data can therefore be summarized as a collection of $N$ independent 
two--dimensional variables
\[ \{x_i, y_i\},\; \text{ with } y_i \sim \text{Poiss}(\mu_i),\; i=1,\dots,N\]
where $\mu_i$ is the unknown parent mean of the Poisson distribution
of the counts $y_i$, collected in a fixed 
range $x_i \pm \Delta x_i/2$ of the independent variable.

In general, the relationship between the dependent and independent variables is
of the form $y = f(x)$, where $f(x)$ is an analytical function
with $m$ adjustable parameters $a_j$, $j=1,\dots,m$. 
The likelihood of the data with the model $f(x)$ is given by

\[ \mathcal{L} = \prod_{i=1}^N \dfrac{e^{-\mu_i} \mu_i^{y_i}}{y_i!} \]

where the adjustable parameters in $f(x)$ are
optimized so that the likelihood $\mathcal{L}$ is maximised, for the
given dataset. Instead of maximizing $\mathcal{L}$ directly, it is convenient
to minimize the function
\begin{equation}
C \equiv -2 \ln \mathcal{L} -B = 2 \sum_{i=1}^N \left( \mu_i - y_i + y_i \ln (y_i/\mu_i) \right) = \sum_{i=1}^N C_i
\label{eq:cstat}
\end{equation}
where
\[ B = 2 \sum_{i=1}^N \left( y_i -y_i \ln y_i + \ln y_i! \right)\]
is a quantity that is independent of the model, and
\[ C_i = 2 \left( \mu_i -y_i+y_i  \ln (y_i/\mu_i) \right). \] 
The statistic $C$ is known as the \emph{Cash} or \cstat, originally proposed 
by \cite{cash1979} and \cite{cousins1984} to model and analyze X--ray observations
of astronomical sources.
For a model with $m$ free parameters, the minimization of the \cstat\ 
yields $m$ equations, in general non--linear, that need to be solved to
obtain the maximum--likelihood best--fit estimates of the $m$ model parameters. 

\subsection{Linear models}

The linear model is a simple
and commonly used relationship between two variables. The customary parameterization
of a linear relationship between two variables $x$ and $y$ is
$ f(x) = a + b x$,
with $a$ and $b$ as the two adjustable parameters.
Another convenient and equivalent parameterization, 
suggested by \cite{scargle2013}, is of the form
\begin{equation}
f(x)=\lambda (1+a(x-x_A))
\label{eq:yscargle}
\end{equation}
with $\lambda$  and $a$ as the two adjustable parameters, 
and $x_A$ a fixed fiducial value of the $x$ variable. As will be
shown, this parameterization is convenient 
when taking the derivative of the terms $\ln f(x_i)$ in Equation~\ref{eq:cstat},
since it leads to a separation between the two parameters $\lambda$ and $a$.
This is the parameterization used in this paper~\footnote{A  maximum--likelihood
solution for the standard form of the linear model with the \cstat\ is reported
in \cite{bonamente2017book}. It leads to a set of two non--linear coupled equations,
whose numerical solution can be challenging.}.

The continuous function $f(x)$ is related to the parent mean $\mu_i$ of each Poisson variable
$y_i$ via an integral over the length of the $i$--th bin,
\begin{equation}
 y(x_i)= \mu_i = \int_{x_i-\Delta x_i/2}^{x_i+\Delta x_i/2} f(x) dx=f(x_i) \Delta x_i,
\label{eq:mui}
\end{equation}
where $y(x_i)$ is a step--wise function describing the Poisson mean 
for each bin, and the last equality applies because of the linearity of $f(x)$. 
It is therefore recognized that $f(x)$ is a non--negative {\blue density function} 
in units of \emph{counts per unit $x$} {\blue (i.e., not just \emph{counts})}, while its integral
over a range $\Delta x_i$ is the predicted Poisson mean $y(x_i)= \mu_i$ in that bin, 
in units of \emph{counts}. Therefore, the two functions $f(x)$ and $y(x_i)$ will
vary from each other according to the size of the bins, which is allowed to
be non--uniform. 
It is important to
stress that the function $f(x)$ must be non--negative, since it would not be meaningful
to have a Poisson variable with a negative parent mean.
The requirement  $f(x)\geq 0$ places a number of constraints on the solutions of the
maximum--likelihood equations that
 are discussed in Section~\ref{sec:acceptability}.

{\blue \subsection{Generalized linear models and other considerations for count data} 
\label{sec:GLM}
Integer--valued count data of the type considered in
this paper can be modelled with alternative statistics
that afford more flexibility than the single--parameter Poisson
distribution, in particular
with regards to over-- or under--dispersion of the data
\citep[e.g.,][]{bonat2018,haselimashhadi2018,sellers2010,shmueli2005}.
Moreover, the regression 
with one or several independent
variables may often require more complex models than a simple linear
model.
Generalized linear models and vector--generalized linear models
provide a comprehensive and flexible
framework for the regression of data, including a natural way to account
for non--negative Poisson means via a suitable link function
that relates the Poisson mean to the data and model parameters \citep[e.g.][]{nelder1972,mccullagh1989,
yee1996,yee2015,dobson2018}.
Within the context of generalized linear models, a convenient
link function would be in the form of $\log \mu = a+b x$, where 
the logarithm of the Poisson mean $\mu$, instead of the mean itself, is
modelled with a linear function \cite[as in, e.g.][]{elsayyad1973}. 
Such log--linear
models would ensure that the Poisson mean is always positive
(e.g., see chapter 6 of \citep{mccullagh1989}).

There are two main reasons for the present investigation of a simple linear regression using 
the standard Poisson distribution, instead of 
more versatile models or distributions. First is 
the goal to
obtain an analytical, and therefore computationally efficient,
solution for the best--fit parameters of the linear model for count data.
Currently, the only available analytical solution for the maximum--likelihood 
fit of the linear model
is for Gaussian--distributed variables, and this method is not accurate for low--count 
integer--valued data \citep[e.g.][]{bonamente2020}.
Second, there may be scientific reasons to prefer a simple linear
model versus, e.g., a log--linear model or other 
more complex models, particularly
when the data or an underlying parent model suggest
a direct linear relation between the dependent and independent variables
\citep[e.g.][]{mock2011, valenti2016}. 
The combination of these practical and scientific reasons
make it interesting to seek an analytical solution of the basic
linear regression with Poisson data.
\rm}

\section{Maximum likelihood solutions for the parameters of the  linear model}
\label{sec:ML}

\subsection{The \cstat\ for the linear model}
The linear model of Equation~\ref{eq:yscargle} is illustrated in Figure~\ref{fig:yscargle}.
To evaluate the \cstat\ of Equation~\ref{eq:cstat} with the linear model
of Equation~\ref{eq:yscargle}, start with

\begin{figure}[!t]
\centering
\includegraphics[width=5in]{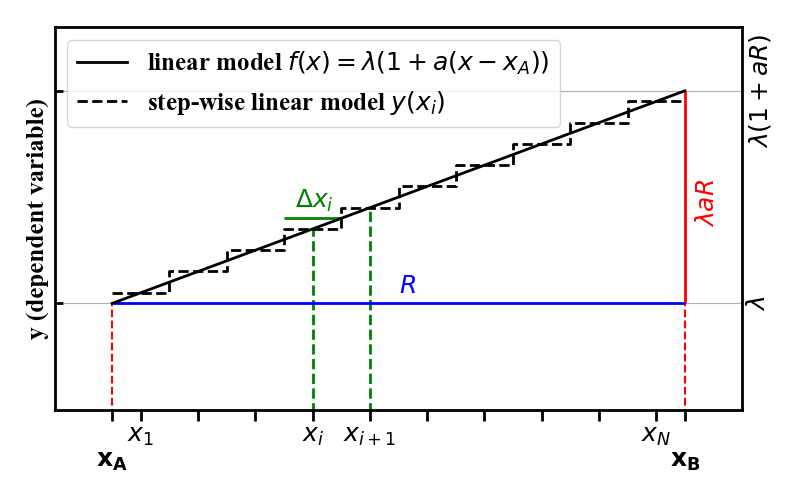}
\caption{Linear model according to Equation~\ref{eq:yscargle}. 
In this illustration, the functions $f(x)$, in units
of counts per unit $x$,
and $y(x_i)$, in units of total counts in the bin, follow one another closely because 
a bin size value of $\Delta x=1$ was used.}
\label{fig:yscargle}
\end{figure}

\begin{equation}
\sum_{i=1}^N \mu_i = \int_{x_A}^{x_B} f(x) dx = 
\lambda R \left(1+a \dfrac{R}{2}\right)
\label{eq:summu}
\end{equation}
where $R=x_B-x_A$ is the range of the $x$ variable, and it is assumed that
the data covers the {\blue entire} range $R$. {\blue There are situations
where measurements of the dependent variable $y$
are missing for certain intervals of the independent variable $x$, for example
because measurements are
not possible or because they are ignored in the analysis. 
An interval of the independent variable where data are missing, 
or are otherwise not used in the analysis,
will be referred to as a \emph{gap}}. If the data contain gaps in the 
$x$ variable, the limits of integration of Equation~\ref{eq:summu} will change.
Section~\ref{sec:gap} describes the  simple modification required to analyze
data that contain {\blue such} gaps.
 
The second term of the \cstat\ is simply 
\[ \sum_{i=1}^N y_i = M, \]
where $M$ is the total number of counts, and the final term is
\[ \sum_{i=1}^N  (y_i \ln (y_i) - y_i \ln \mu_i)=
	\sum_{i=1}^N  y_i \ln (y_i) - \sum_{i=1}^N  y_i \ln f(x_i) - \sum_{i=1}^N  y_i \ln \Delta x_i,  \]
where Equation~\ref{eq:mui} was used.
The \cstat\ for the linear model of Equation~\ref{eq:yscargle} is therefore
\begin{dmath}
C = 2\lambda R\left(1+\frac{aR}{2}\right) -2 M \ln \lambda 
	-2\sum_{i=1}^N  y_i \ln (1+a(x_i-x_A)) +D
\label{eq:cstatscargle}
\end{dmath}
where
\[
D = \left(2\sum_{i=1}^N  y_i \ln y_i
-2M -2\sum_{i=1}^N  y_i \ln \Delta x_i\right),
\]
is a  model--independent term that  does not have an effect 
in the subsequent minimization of the statistic.  Note that the
binning of the data is not required to be uniform, as will be
illustrated in subsequent examples. 

\subsection{Equations for the maximum--likelihood solutions}
Minimization of the \cstat\ in Equation~\ref{eq:cstatscargle} is obtained via

\[ \frac{\partial C}{\partial \lambda} = 2 R + a R^2 -2\frac{M}{\lambda} = 0, \]

which leads to 
\begin{equation}
\lambda(a) = \frac{M}{R\left(1+a \dfrac{R}{2}\right)}
\label{eq:lambda}
\end{equation}

and 

\[ \frac{\partial C}{\partial a} = \lambda R^2 -2 \sum_{i=1}^N y_i \dfrac{(x_i-x_A)}{1+a(x_i-x_A)} =0.\]
Substituting Equation~\ref{eq:lambda} to eliminate $\lambda$ leads to
\[ \dfrac{M R}{1+a\dfrac{R}{2}} - 2 \sum_{i=1}^N y_i \dfrac{(x_i-x_A)}{1+a(x_i-x_A)} =0 \]
and finally to 
\begin{equation}
a = -\dfrac{2}{R} + \dfrac{M}{\sum_{i=1}^N y_i \dfrac{(x_i-x_A)}{1+a(x_i-x_A)}}, 
\label{eq:a}
\end{equation}
which is the equation to solve for the values of the $a$ parameter.
Equation~\ref{eq:a} may be rearranged as
\[
1 + a\dfrac{R}{2} - \dfrac{MR}{2 g(a)} = 1+\dfrac{R}{2} \left(a -\dfrac{M}{g(a)} \right)
=0;\]
It is thus convenient to define
\begin{equation}
F(a) \equiv  1+\dfrac{R}{2} \left(a -\dfrac{M}{g(a)} \right)
\label{eq:Fa}
\end{equation}
as the function whose zeros are solutions of Equation~\ref{eq:a},
with $g(a)$ defined as
\begin{equation}
g(a)\equiv \sum_{i=1}^N y_i \dfrac{(x_i-x_A)}{1+a(x_i-x_A)}.
\label{eq:ga}
\end{equation} 
The problem of finding solutions for the parameter $a$ has therefore been
cast as finding the zeros of a function $F(a)$, which uses the function
$g(a)$ of Equation~\ref{eq:ga}.
One of the key properties to find the zeros of $F(a)$ is that
the zeros of $g(a)$ are points of singularity for $F(a)$, as will be 
shown in detail in Section~\ref{sec:analytical}.

In summary, $F(a)=0$ and Equation~\ref{eq:lambda} are the 
two equations to solve to find the maximum--likelihood 
estimators $a$ and $\lambda$ of the linear model of Equation~\ref{eq:yscargle}. 
The two equations are
uncoupled, and therefore the burden is limited to finding a
solution of
$F(a)=0$.~\footnote{
\blue It is useful to point out that a log--linear  model,
as obtained for example using a logarithmic link function within
the context of generalized linear models (see Section~\ref{sec:GLM}), 
would have led to coupled equations
involving the exponential of the parameters, in place of
Equations~\ref{eq:lambda} and \ref{eq:a}.
}
Then, Equation~\ref{eq:lambda} is used to find $\lambda=\lambda(a)$.
Notice that, in deriving Equation~\ref{eq:lambda} and \ref{eq:a}, no constraints were
enforced to ensure that the best--fit model 
is non--negative, which is necessary
for the applicability of the Poisson statistics and for the calculation \cmin.
These constraints are presented in Section~\ref{sec:acceptability}.

\section{Analytical properties of the maximum--likelihood solution of the linear model}
\label{sec:analytical}

The maximum--likelihood estimate of the parameters $\lambda$ and $a$ are obtained 
by first finding the zeros of the function $F(a)$ defined by Equation~\ref{eq:Fa}. 
For this purpose, it is necessary to establish a 
few analytical properties of the functions $F(a)$ and $g(a)$.
These properties will be used to study the location and properties 
of the zeros of the function $F(a)$.
It is necessary to discuss explicitly the simple case of data with $M=1$ count
before presenting general results for $M\geq 2$. The case of $M=0$ counts is not
interesting, since it represents a dataset with no positive measurements
throughout the range of the independent variable.

\subsection{Case of $M=1$} 
When $M=1$, or just one event ($y_j=1$) 
recorded in a bin centered at $x_j$, one obtains 
\[ g(a) = \dfrac{x_j-x_A}{1+a(x_j-x_A)}\] 
leading to
\[
 F(a)= 1+ a\dfrac{R}{2} - \dfrac{R (1 + a(x_j-x_A))}{2(x_j-x_A)} = 1 - \dfrac{R}{2(x_j-x_A)}
\]
which is constant independent of $a$.  The conclusion is that
 $F(a)=0$ has no solutions when the data have only one count, 
and it is therefore 
not possible to find a maximum--likelihood solution for the linear model
with $M=1$. A simple interpretation for this finding is that
it is not possible to constrain a two--parameter model with just
one non--null data point. Further discussion is provided in Section~\ref{sec:discussion}.

\subsection{General case of $M \geq 2$}
 It is possible to find certain properties of $g(a)$ and $F(a)$ that apply
in general. The properties will lead to a general criterion to identify solutions
of $F(a)=0$.

\begin{property}[Properties of the function $g(a)$]
\label{prop:ga}
The function $g(a)$, according to Equation~\ref{eq:ga}, is the sum
of $n \leq M$ terms 
of type 
\[g_j(a) = y_j \dfrac{(x_j-x_A)}{1+a(x_j-x_A)},\]
where $n$ is the number of bins with $y_j \neq 0$, and $j=1,\dots,n$ represents
the bins with non--null counts; 
when bins have no more
than one count, then $n=M$.
Therefore, the function $g(a)$ has $n$ points of singularity 
\[a_j=-\dfrac{1}{(x_j-x_A)} \]
that fall in the range $(-2/\Delta x_1, -1/(R-\Delta x_{N}/2)$.
Near the points of singularity, 
\begin{align*}
\lim_{a \to a_j^-} g(a) = - \infty \\
\lim_{a \to a_j^+} g(a) = + \infty.
\end{align*}
It is also
immediate to see that 
$g'(a)<0$ for all points 
where the function is continuous.
Therefore, the function $g(a)$ decreases monotonically from $+\infty$ to $-\infty$
between two consecutive points of singularity.
Moreover, the asymptotic limits are
\[\lim_{a \to \pm \infty} g(a)=0.\]
As a result, the function $g(a)$ has  $n-1$ zeros, each between consecutive 
points of singularity, as illustrated in Example~\ref{ex:M=2} and 
Figures~\ref{fig:M=2} and \ref{fig:M=3}. The zeros of $g(a)$ are  
points of singularity for $F(a)$.

In particular, when $M=2$ with $y_j=y_k=1$ and no counts in any of the other bins,
then the zero of $g(a)$ can be calculated analytically as
\begin{equation*}
 a_c = - \dfrac{1}{2} \left(\dfrac{1}{x_j-x_A} + \dfrac{1}{x_k-x_A}\right) < 0 \;\; \text{ (case of $M=2$ only)}.
\label{eq:ac}
\end{equation*}
For the general case of $M>2$, the zeros of $g(a)$ must be calculated numerically,
as explained in the following.
\end{property}

\begin{property}[Behavior of $F(a)$ near the points of singularity]
\label{prop:FaS} 
Since $g(a)$ is continuous with $g'(a)<0$ between its $n$ points of singularity,
$g(a) <0$ immediately to the left of the singularity and $g(a)>0$
immediately to the right. With $a_s$ any of the $n-1$ points of singularity of $F(a)$,
this implies that
\begin{equation*}
 \begin{cases}
\lim\limits_{a\to a_s^-} F(a) = -\infty\\
\lim\limits_{a\to a_s^+} F(a) = +\infty.
\end{cases}
\end{equation*}
\end{property}

\begin{property}[Asymptotic limit of $F(a)$] 
\label{prop:Fae}
	The asymptotic limit of $F(a)$ at $\pm \infty$, {\blue defined
	as
	\[ F_{\infty} \equiv \lim_{a\to \pm \infty} F(a),\]}
	
	can be evaluated 
via the \emph{De L'Hospital} rule for the associated function $a-M/g(a)$,
which is an indeterminate form of the type $0/0$: 
\[
\lim_{a\to \pm \infty} \left(\dfrac{a g(a) -M}{g(a)}\right) = 
%
\lim_{a\to \pm \infty} \left(\dfrac{g(a) + a g'(a)}{g'(a)}\right)
= - \dfrac{1}{M} \sum_{i=1}^N \dfrac{y_i}{x_i-x_A} = - \dfrac{1}{M} \sum_{j=1}^M \dfrac{1}{x_j-x_A}.
\]
This property can be proven with a few steps algebra, 
by turning the sum over the $N$ bins to an
equivalent sum over the individual $M$ counts.
Therefore the asymptotic limit of $F(a)$ becomes
\begin{equation}
\lim_{a\to \pm \infty}  F(a) =  F_{\infty} =  1 - \dfrac{R}{2} 
	\left( \dfrac{1}{M} \sum_{i=1}^N \dfrac{y_i}{x_i-x_A} \right). 
\label{eq:Flimitinfinity2}
\end{equation}
\end{property}

\begin{figure}[!t]
\includegraphics[width=3.2in]{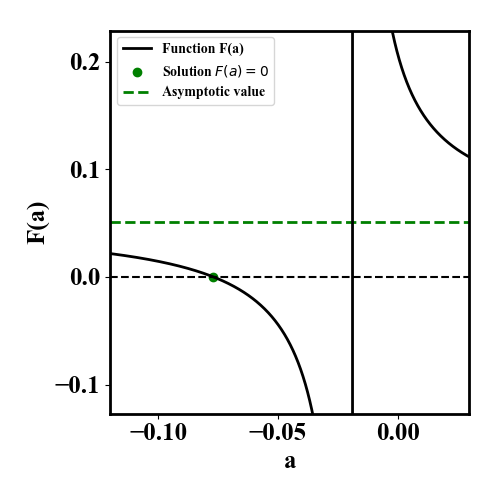}
\hspace{-0.5cm}
\includegraphics[width=3.2in]{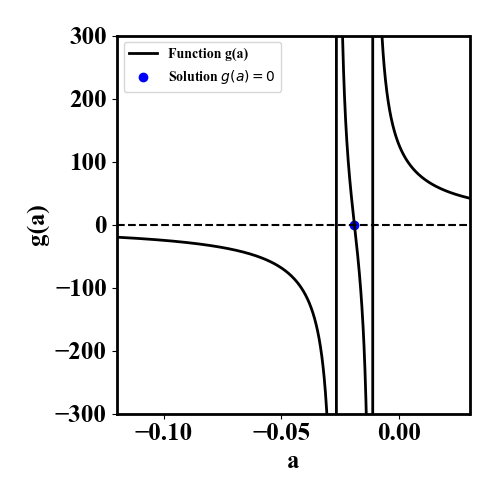}
\caption{(Left:) Function $F(a)$ for a representative
data set with $M=2$, with 100 data points $x_i=i-0.5$, $i=1,\dots,100$,
and $y_{38}=y_{89}=1$.
(Right:) Function $g(a)$ for the same data set.}
\label{fig:M=2}
\end{figure}

\begin{property}[Sign of $F'(a)$]
\label{prop:Faprimee}
The derivative $F'(a)$ is calculated according to 
\begin{equation}
F'(a) = \dfrac{R}{2} +\dfrac{MR}{2}\dfrac{g'(a)}{g^2(a)}.
\label{eq:Faprime}
\end{equation}
For $M \geq 2$,
\[ \dfrac{g'(a)}{g^2(a)} = - \dfrac{z_1^2+ \dots + z_M^2}{(z_1+\dots+z_M)^2}, \text{ with } z_j\equiv\dfrac{x_j-x_A}{1+a(x_j-x_A)} 
\]
in which the sum is over all individual events, with some $z_j$ identical
to each other if there are bins with more than one count.
Using the \emph{Cauchy--Schwartz} inequality
\[ \dfrac{1}{M} \left( \sum_{j=1}^M z_j \right)^2 \leq \sum_{j=1}^M z_j^2 
\]
leads to 
\begin{equation}
	\dfrac{g'{\blue (a)}}{g^2(a)} \leq -\dfrac{1}{M}.
\label{eq:cauchy-schwartz}
\end{equation}
Finally, using Equation~\ref{eq:cauchy-schwartz} into Equation~\ref{eq:Faprime} leads to
the conclusion that $F'(a) \leq 0$ for all points of continuity of $F(a)$. 
\end{property}

These properties of the function $F(a)$ are also illustrated
in Example~\ref{ex:M=2} and
Figures~\ref{fig:M=2} and \ref{fig:M=3}.
The properties of $F(a)$ and $g(a)$ can be used to state a general
criterion to locate all solutions of the equation $F(a)=0$. 

\begin{figure}[!t]
\includegraphics[width=3.0in]{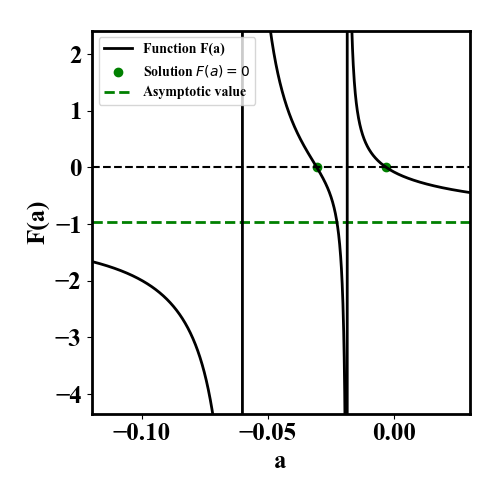}
\includegraphics[width=3.0in]{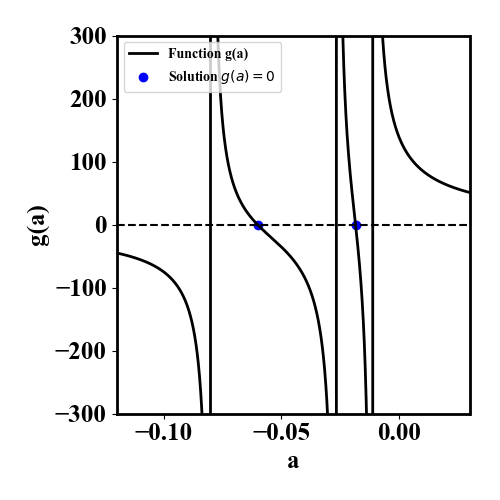}
\caption{(Left:) Function $F(a)$ for a representative
data set with $M=3$, with 100 data points $x_i=i-0.5$, $i=1,\dots,100$,
and $y_{13}=y_{38}=y_{89}=1$.
(Right:) Function $f(a)$ for the same data set.}
\label{fig:M=3}
\end{figure}

\begin{lemma}[Location of zeros of $F(a)$]
The function $F(a)$ has 
 $n-1$ zeros, where $n\leq M$ is the number of bins with non--null counts.
Of these, $n-2$ zeros are found between the $n-1$ points of singularity of $F(a)$, also zeros of $g(a)$.
The remaining zero is found either
to the left of the smallest point of singularity, if the asymptotic limit is $F_{\infty}>0$, 
or  to the right of the largest singularity, if $F_{\infty}<0$.
\end{lemma}
\begin{proof}
The result is a direct consequence of the presence of $n$ points 
of singularity for $F(a)$ (property~\ref{prop:ga}), the negative sign 
of $F'(a)$ between points of singularity (property~\ref{prop:Faprimee})
and the asymptotic limit of $F(a)$ at the points of singularity (property~\ref{prop:FaS}).
\end{proof}

Properties of $g(a)$ and $F(a)$ are illustrated in the following
example, which examines the behavior of the two functions
for two simple datasets with $M=2$ and $M=3$.

\begin{example}[Two datasets with $M=2$ and $M=3$]
\label{ex:M=2}
\label{ex:M=3}
Two sample datasets with $M=2$ and $M=3$ are shown respectively in Figure~\ref{fig:M=2}
and \ref{fig:M=3}.
For the $M=2$ dataset, with $n=2$ bins with non--null counts, the function $F(a)$
has just one point of discontinuity for $F(a)$, also the zeros of $g(a)$.
This point of discontinutiy divides the domain of $a$ into
two intervals, with $F(a)$ monotonically decreasing within these
intervals, as shown in Figure~\ref{fig:M=2}.
For the $M=3$ dataset, with $n=3$ bins with non--null counts, $F(a)$ has $n-1=2$ points
of discontinuity corresponding to the two zeros of $g(a)$, which were in turn
found between the $n=3$ points of singularity of $g(a)$, as shown in Figure~\ref{fig:M=3}.
\end{example}
The $n-1$ zeros of $F(a)$ are all possible solutions
of the maximum--likelihood method for the linear model.
The following section discusses if these solutions are acceptable, in the
sense that they produce a model $f(x)$ that is non--negative 
throughout the domain of the $x$ variable.

\section{Acceptable solutions for the best--fit parameters of the linear model}
\label{sec:acceptability}
In Section~\ref{sec:analytical} it was shown that there are several possible
solutions for the maximum--likelihood parameters of the linear model of Equation~\ref{eq:yscargle}.
In particular, data with $M$ total counts, distributed over $n\leq M$ 
of the $N$ available bins,
have $n-1$ possible values of $a$ that are a solution
of the maximum--likelihood equation $F(a)=0$, with 
the corresponding value of $\lambda$ provided
by Equation~\ref{eq:lambda}.
This section addresses the additional requirement that a model be non--negative
in all bins, i.e., that a solution be acceptable, so that Poisson statistics apply and the \cstat\ can be calculated. 
\begin{definition}[Acceptable solution of $F(a)=0$]
A solution of the $F(a)=0$ equation is said to be \emph{acceptable}
if it leads to a non--negative model throughout the support of the
independent variable. Specifically,
the function  must satisfy the condition that
$y(x_i) \geq 0, \; i=1,\dots,N$, so that the
parent mean of the Poisson distributions is always non--negative.
\end{definition}
It will be shown in this section that there is 
\emph{at most} one acceptable solution for any Poisson
data set. Cases without an acceptable solution will be examined in 
Section~\ref{sec:generalized}, where a simple 
generalization of the linear model is provided that
ensure \emph{one and only one} acceptable solution for the fit of any data set
to a linear model.

\subsection{General conditions for acceptability}
 Given that the model is linear,
the condition of acceptability is satisfied by simply requiring that 
the Poisson mean for the first and last bins,
$y(x_1)$ and $y(N)$, are both non--negative, 
\[
\begin{cases}
\lambda \cdot (1+a \Delta x_1/2) \geq 0 \\
\lambda \cdot (1+a(R-\Delta x_{N}/2)) \geq 0.
\end{cases}
\] 
Notice how the model $f(x)$ may still become negative in a portion
of either the first or the last bin, but the linearity of the model simply requires that $f(x)$ at
the mid--point of the bin be non--negative, in order to ensure that the Poisson
mean for the bin is non--negative.

Substituting Equation~\ref{eq:lambda}, the conditions become a function of $a$ alone,
\begin{equation}
\begin{cases}
 \dfrac{M}{R\left(1+a\dfrac{R}{2}\right)} (1+a \Delta x_1/2) \geq 0 \\
\dfrac{M}{R\left(1+a\dfrac{R}{2}\right)} (1+a(R-\Delta x_{N}/2)) \geq 0,
\end{cases}
\label{eq:conditionsAcceptability}
\end{equation}
This equation can be used to find a range of the variable $a$ that contains
acceptable solutions. Therefore, a solution of $F(a)=0$ is acceptable if and only if it
satisfies Equations~\ref{eq:conditionsAcceptability}.
This property leads to the following result regarding acceptable solutions:

\begin{lemma}[Necessary and sufficient condition for the acceptability of a solution of
$F(a)=0$]
A solution of $F(a)=0$ is acceptable if and only if it is found
outside of the interval 
\[ \left(-\dfrac{2}{\Delta x_1}, -\dfrac{1}{R-\Delta x_{N}/2} \right).\]
\label{lemma:acceptability}
\end{lemma}
\begin{proof}
The conditions of Equations~\ref{eq:conditionsAcceptability} can be used to find
values of the variable $a$ that are acceptable solutions of $F(a)=0$.
For $a < -2/R$, i.e., when the denominators in Equation~\ref{eq:conditionsAcceptability}
are negative, the two conditions are satisfied when $a < -2/\Delta x_1$,
%
since $\Delta x_1 < R-\Delta x_{N}/2$.
Likewise, for $a> -2/R$, the two conditions 
are satisfied when $a \geq -1/(R-\Delta x_{N}/2)$.
Therefore, acceptable solutions can be found in the range
\begin{equation}
\begin{cases}
\begin{aligned} &a < -2/\Delta x_1 &\text{ (and thus $\lambda <0$)}\\
&a> -1/(R-\Delta x_{N}/2) &\text{ (and thus $\lambda >0$).}
\end{aligned}
\end{cases}
\end{equation}
\end{proof}
Solutions with $a \in (-2/\Delta x_1, -1/(R-\Delta x_{N}/2)$ lead to a model that becomes negative in some of the
bins, and therefore they are not acceptable. 
Figure~\ref{fig:lambda} shows the function $\lambda(a)$ and illustrates the
range of acceptable values for the parameters.
\begin{figure}[th]
\centering
\includegraphics[width=4.5in]{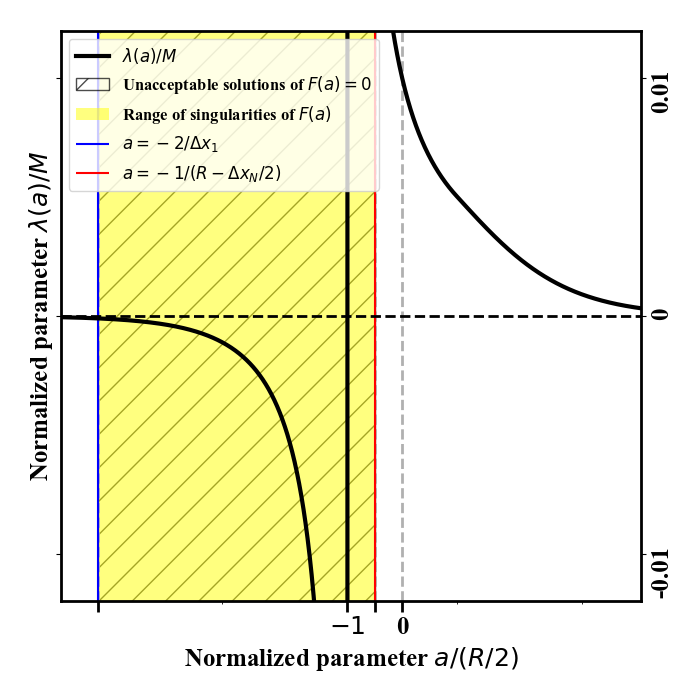}
\caption{(Left:) Parameter $\lambda$ as a function of the parameter $a$,
and range of acceptability of $a$.
For values of $-2/\Delta x_{1} < a< -1/(R-\Delta x_{N}/2)$ the parameters $(a,\lambda)$ result
in a linear model that becomes negative in some of the bin, and therefore not acceptable
for use with the Poisson distribution.   
The smallest and largest points of singularity
 of $g(a)$ also correspond to the boundaries of this range, 
according to Property~\ref{prop:ga}. The zeros of $g(a)$, also points of singularity
for $F(a)$, are therefore inside this range.
 The
$x$ axis was plotted with the \texttt{symlog}
option that allows a near--logarithmic scaling
across a value of zero.}
\label{fig:lambda}
\end{figure}

\begin{example}[$M=2$ data with no acceptable solution]
The acceptability of the maximum--likelihood solution is illustrated
with the data used for Figure~\ref{fig:M=2}.
The only singularity of the $F(a)$ function is 
\[ a_c = -\dfrac{1}{2}\left(\dfrac{x_{38}}{1+a x_{38}} + \dfrac{x_{89}}{1+a x_{89}}\right) = -0.019
\]
which is $a_c > -2/R$, and with a positive asymptotic value of $F_{\infty}=0.051$, as
shown in Figure~\ref{fig:M=2}. 
The $F(a)=0$ solution is  $a=-0.077$, which falls in the range
of unacceptable solutions. In fact, the corresponding $\lambda(a)=-0.007$ results in a best--fit
model that is negative in some of the initial bins, e.g., $y(x_1)=\lambda ( 1 + a \Delta x/2)=
-0.0067$.
This best--fit model cannot be used
to calculate the goodness--of--fit \cstat, and therefore cannot be 
accepted as a maximum--likelihood solution. 
\end{example}

\subsection{General method to locate acceptable solutions}

\begin{lemma}[Necessary condition for the acceptability of $F(a)=0$ solutions]
Solutions of $F(a)=0$ within points of singularity of $F(a)$ are always unacceptable.
\label{lemma:necessary}
\end{lemma}
\begin{proof}
This condition applies to data with $M> 2$ and $n>2$ unique bins with non--zero counts,  
so that $F(a)$ has $n-1\geq 2$ points of singularity.
In this case, $n-2$ of the $n-1$ solutions of $F(a)=0$ 
are found between the $n-1$ points of singularity
of $F(a)$, which are the zeros of the function $g(a)$. 
According to property~\ref{prop:ga}, 
the zeros of $g(a)$
are located between the $n$ points of singularity of $g(a)$, given by 
\begin{equation}
 a_j = -\dfrac{1}{x_j-x_A} \in \left( -\dfrac{2}{\Delta x_1}, -\dfrac{1}{R-\Delta x_{N}/2}\right),
\label{eq:gaSingularity}
\end{equation}
where $x_j$ is the coordinate of each of the $n\leq M$ unique bins where non--zero
counts are recorded. 
According to lemma~\ref{lemma:acceptability}, these $n-2$ solutions of $F(a)=0$
 fall in the interval on unacceptability.
\end{proof}

Lemma~\ref{lemma:necessary} states 
that all zeros of $F(a)$ that are within points of singularity may be discareded
as unacceptable. The only possibility for an acceptable solution
is the zero that is located
outside of the range of the points of singularity,
although such zero 
is not guaranteed to be acceptable. Accordingly, the following definition is made:

\begin{definition}[External solution of $F(a)=0$]
A solution of $F(a)=0$ is said to be \emph{external} if it falls outside of
the range of the points of singularity of $F(a)$.
\end{definition}
An external solution is therefore found either to the left of the first
point of singularity of $F(a)$, if the asymptotic value $F_{\infty} > 0$, or to the right
of the last, if $F_{\infty} < 0$. For $M=2$, this is the only solution of $F(a)=0$, and
the point of singularity of $F(a)$ is calculated according
to the equation provided at the end of Property~\ref{prop:ga}.

\begin{lemma}[Necessary and sufficient conditions for the acceptability of an external $F(a)=0$ solutions]
An external solution of $F(a)=0$ is acceptable if
and only if the following conditions are met, according to the 
sign of the asymptotic value $F_{\infty}$:
\begin{equation}
\begin{cases}
\begin{aligned}
&F\left(-\dfrac{2}{\Delta x_1}\right) <0,\; &\textrm{ if } F_{\infty} > 0\\
&F\left(-\dfrac{1}{R-\Delta x_{N}/2}\right) >0,\; &\textrm{ if } F_{\infty} < 0.
\end{aligned}
\end{cases}
\label{eq:necessarySufficient}
\end{equation}
\label{lemma:necessarySufficient}
\end{lemma}
\begin{proof}
This property is simply based on the continuity of the 
function $F(a)$ between points of singularity,
and on lemma~\ref{lemma:acceptability}, which established that
acceptable values of the parameter $a$ are outside of the interval $(-2/\Delta x_1, -1/(R-\Delta x_{N}/2)$.

(a) if $F_{\infty} > 0$, the solution
of $F(a)=0$ is to the left of the point of singularity. Given that $F'(a)<0$, the solution
will fall in the range of acceptability, viz., $a< -2/R$, if
$F(-2/\Delta x_1) <0$.

(b) Likewise, if $F_{\infty} < 0$, the solution is to the right of the point of singularity,
and the solution is acceptable if
$F(-1/(R-\Delta x_{N}/2))>0$.

The condition is also necessary. In fact, if Equation~\ref{eq:necessarySufficient}
is not satisfied, e.g., $F(-2/\Delta x_1)<0$ for $F_{\infty} > 0$, then the zero will be in the
region of unacceptability.
\end{proof}

This necessary and sufficient condition can be immediately applied
to data that have non--zero counts, and thus points of singularity of $g(a)$, at the extremes
of the range.

\begin{corollary}[Sufficient conditions for data with non--zero counts in
first or last bin]
\label{corollary:firstLast}
If a data set with $M\geq2$ satisfies either of the two conditions
\begin{equation}
\begin{cases}
y_1\geq1 \textrm{ and } F_{\infty}>0\\
y_{N}\geq1 \textrm{ and } F_{\infty}<0
\end{cases}
\end{equation}
then the external solution of the $F(a)=0$ equation is acceptable.
\end{corollary}
\begin{proof}
According to Equation~\ref{eq:Fa},  $F(a)=1+a\cdot R/2$ at points of singularity
for $g(a)$.
A non-zero count in the first bin leads to a singularity of $g(a)$ at $a_j=-\Delta x_1/2$,
and therefore $F(-2/\Delta x_1)<0$.
Therefore, according to lemma~\ref{lemma:necessarySufficient}, 
if $F_{\infty}>0$, there is an acceptable
solution to the left of $-2/\Delta x_1$.
Similar considerations are applicable to the case of a non-zero count in the last
bin, where a singularity of $g(a)$ occurs instead at $a_j=-1/(R-\Delta x_N/2)$,
where $F(-1/(R-\Delta x_N/2))>0$. In this case, if $F_{\infty}<0$, the external
solution of $F(a)=0$ to the right of the last singularity of $g(a)$ is acceptable,
again according to lemma~\ref{lemma:necessarySufficient}.
\end{proof}

Notice how corollary~\ref{corollary:firstLast} does not ensure an acceptable
solution simply if either the last or first bin have non--null counts.
In fact, the presence of an acceptable solution is conditioned also on the sign of
$F_{\infty}$. For example, a data set with a non--null last bin but with a positive
$F_{\infty}$ will not have an acceptable solution to the right of the last singularity.
Finally, the two earlier lemmas can be used to state the uniqueness
of the maximum--likelihood solution for the linear model.
\begin{lemma}[Uniqueness of an acceptable solution of $F(a)=0$]
If there is an acceptable solution of $F(a)=0$, this solution is unique.
\label{lemma:uniqueness}
\end{lemma}
\begin{proof}
This property is an immediate consequence of the fact that, of the $n-1$ solutions
of $F(a)=0$, the $n-2$ solutions within points of singularity cannot
be acceptable, as per Lemma~\ref{lemma:necessary}.
Moreover, the remaining solution may be acceptable, according to Lemma~\ref{lemma:necessarySufficient}.
\end{proof}

\begin{example}[Example of data with $M=5$ and an acceptable solution]
\label{ex:M=5}
Figure~\ref{fig:M=5a} shows the $F(a)$ and $g(a)$ functions for a
dataset with $M=5$ counts in 5 equally--spaced bins ($x_i=9.5,29.5, 49.5, 69.5, 89.5$),
and therefore $n=5$.
The function $g(a)$ has $n=5$ points of singularity and $n-1=4$ zeros, 
which correspond to the
4 points of singularity of $F(a)$. There are also 4 zeros of the function $F(a)$,
of which 
$n-2=3$ correspond to unacceptable solutions. The asymptotic value
if $F_{\infty} <0$, and therefore the remaining external zero is to the right of the last
singularity. At the end point $a_2$ of the range of acceptability, the function 
is $F(a_2)>0$, and therefore the last zero leads to an acceptable solution.

The data and all models are shown in Figure~\ref{fig:M=5b}. Notice how the 
model corresponding to the solutions of $F(a)=0$ that are not acceptable lead to
a model that becomes negative; these models cannot be used for the \cstat, and need
to be rejected. The only acceptable model is shown as a solid line, and the corresponding
values of the \cstat\ for each bin are shown in the right panel.
\end{example}

\begin{figure}
\includegraphics[width=3.0in]{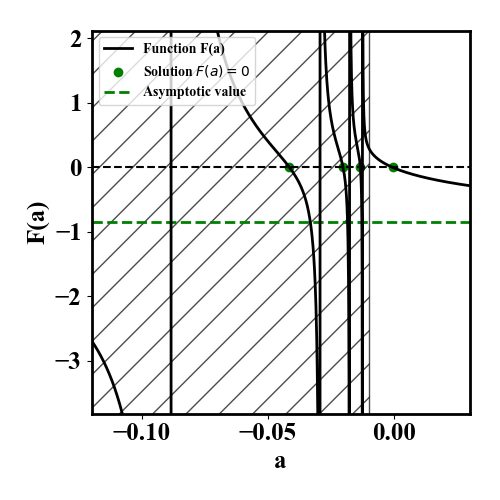}
\includegraphics[width=3.0in]{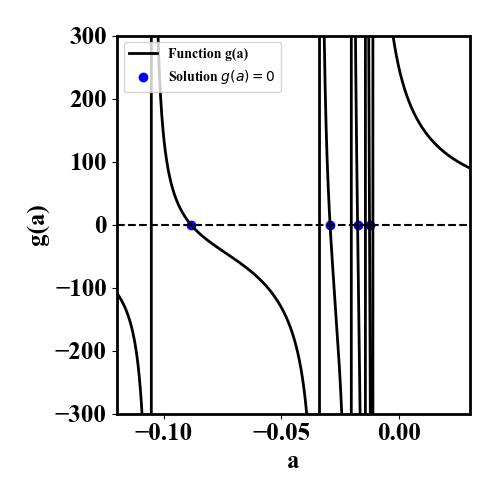}
\caption{Functions $F(a)$ and $g(a)$ for the dataset presented in Example~\ref{ex:M=5}.}
\label{fig:M=5a}
\end{figure}

\begin{figure}
\includegraphics[width=3.2in]{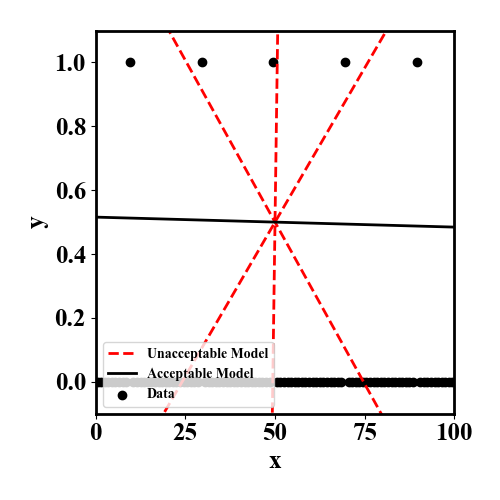}
\includegraphics[width=3.2in]{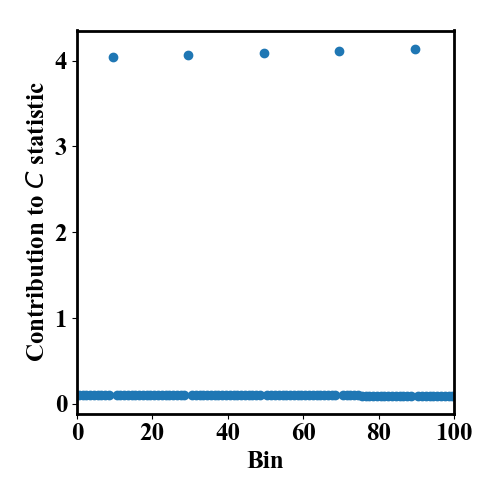}
\caption{(Left) Best--fit linear models for the $M=5$ data presented in Example~\ref{ex:M=5}.
There are 4 solutions of $F(a)=0$, of which the first three lead to models that become
negative somewhere in the $x$ range; the acceptable model corresponds to the largest
solution. (Right) Contributions to \cstat\ for each of the $N=200$ bins,
for a total of $C_{min}=29.96$.}
\label{fig:M=5b}
\end{figure}

The results presented in this section can be summarized by a simple
algorithm that can be used to determine whether there is an
acceptable solution of the equation $F(a)=0$, and to calculate it, when it exists.

\begin{remark}[Algorithm to determine  acceptable best--fit parameters
of Equation~\ref{eq:yscargle}]
\label{remark1}
Consider a dataset with $N$ bins, a range $R$ of the
independent variable between $x_A$ and $x_B$, a total number of integer--valued counts $M$
with a number $n \leq M$ bins with non--null counts.
The existence and value of the best--fit parameters $\{a,\lambda\}$ for the linear model
of Equation~\ref{eq:yscargle} can be determined according to the following steps:
\begin{enumerate}
\item If $n \leq1$, there is no acceptable solution.
\item Calculate the $n\geq 2$ points of singularity for $g(a)$, given analytically by Equation~\ref{eq:gaSingularity}.
\item Numerically calculate the $n-1$ zeros of $g(a)$ between points of singularity. These zeros are
points of singularity for $F(a)$.
\item Calculate the asymptotic value $F_{\infty}$, given analytically by Equation~\ref{eq:Flimitinfinity2}.
\item (Optional) Numerically calculate the $n-2$ zeros of $F(a)$, found between points of singularity of $F(a)$ (also
zeros of $g(a)$. These zeros always lead to unacceptable solutions.
\item Numerically calculate the remaining external zero of $F(a)$, 
either to the left of the 
first point of singularity (if $F_{\infty}>0$),
or to the right of the last singularity (if $F_{\infty}<0$). 
\item Determine the acceptability of the external solution. Two cases are possible:
\begin{itemize}
\item[a.] If this value is \emph{outside} of the range $(-2/\Delta x_1, -1/(R-\Delta x_{N}/2)$, 
then the solution is \emph{acceptable}. The corresponding value of $\lambda$ can be calculate according to Equation~\ref{eq:lambda}.
\item[b.] If this value is \emph{inside} the range $(-2/\Delta x_1, -1/(R-\Delta x_{N}/2)$, 
then the solution is not acceptable. 
The data do not have an acceptable solution with the model of Equation~\ref{eq:yscargle}.
\end{itemize}
\end{enumerate}
The numerical solution of both equations $g(a)=0$ and $F(a)=0$ 
are facilitated by the continuity of the two functions
between the known points of singularity, or 
between the last point of singularity and $\pm \infty$.
An efficient and accurate numerical routine is provided, for example,  by \texttt{python}'s \texttt{root_scalar}, with the 
\texttt{brentq} method. The method requires the specification of an interval, or \texttt{bracket}, where the 
solution is sought. This is either an interval between the two adjacent points of singularity,
or  an open interval either below the first singularity, or above the last singularity. 
For example, a zero of $g(a)$
can be sought in the range  $[ a_j+\epsilon, a_{j+1}-\epsilon]$, where $a_j$ and $a_{j+1}$ are two consecutive
points of singularity of $g(a)$. This bracket
requires a small value $\epsilon$, to be determined according to the separation of the
data points, to ensure that the function $g(a)$ at the two
 extremes of the bracket has opposite signs.
\end{remark}

\subsection{Asymptotic data requirements for acceptable solutions}

This section examines when
data sets with a large number of counts  have an acceptable solution.
It will be shown that, when the counts are distributed
uniformly across the support, data with large $M$ will always have an
acceptable solution.
In general, however, it is possible to find datasets with large $M$ that do not
have an acceptable solution, depending on the distribution of counts.
This observation will lead to a generalization of the simple model
of Equation~\ref{eq:yscargle}, presented in Section~\ref{sec:generalized}.
 First, it is necessary to investigate how
the asymptotic value of $F(a)$ is affected by the distribution of 
the detected counts. 

\begin{property}[Properties of $F_{\infty}$]
\label{prop:Finf}
The asymptotic value of $F(a)$ is given by Equation~\ref{eq:Flimitinfinity2},
and it is negative if
\begin{equation}
\dfrac{1}{M} \sum_{i=1}^N \dfrac{y_i}{x_i-x_A} = \dfrac{1}{M} \sum_{i=1}^M \dfrac{1}{x_i-x_A} > \dfrac{2}{R}
\label{eq:Flimit<0}
\end{equation}
and positive otherwise.
Given that $0 < x_i -x_A < R$, each term $1/(x_i-x_A)$ has a value $<2/R$ if
$x_i$ is above the midpoint of the range, and a value $>2/R$ if $x_i$ is below
the midpoint. The left--hand side of Equation~\ref{eq:Flimit<0}
is the sample mean of the variable $1/(x_i-x_A)$.
\end{property}

It can now  be established that,
for data with a large number of bins and a uniform
distribution of the counts, the asymptotic value of $F(a)$ is negative.
Moreover, when the number of counts $M$ is also large, the external solution
to the right of the last singularity will be acceptable.

\begin{lemma}
\label{lemma:average}
For a large number of bins $N$ and a uniform distribution of counts, 
\[E\left[\dfrac{1}{x_i-x_A}\right]>\dfrac{2}{R},
\] 
	where $E[\;]$ is the expectation based on a parent distribution {\blue for the
	position $x_i$ of the $i$--th count}.
Moreover, when $M$ is large, the asymptotic value of $F(a)$ is negative.
\end{lemma}
\begin{proof}
Assuming bins of uniform width, the range is $R=\Delta x \cdot N$. 
	The distance {\blue of the $i$--th count from the initial
	point of the range is} $x_i - x_A \in (\Delta x/2, R - \Delta x/2)${\blue, and it} can be written as 
\[ 
x_i-x_A  = \dfrac{\Delta x}{2} +f (R-\Delta x),
\]
 where $f \in (0,1)$.
When $N$ is large and the counts are uniformly distributed
in the $N$ bins, it is possible to treat $f$ as 
a continuous and uniformly distributed random variable {\blue in the range $(0,1)$,
thus with unit probability distribution function}. Accordingly, 
the expectation {\blue of the inverse of the distance $x_i-x_A$} can be {\blue approximated} as
\[
	\begin{aligned}
		E\left[\dfrac{1}{x_i- x_A}\right]_{\text{unif}} 
		 \simeq \int_{0}^{1} \dfrac{df}{\Delta x/2+f(R-\Delta x)} = \dfrac{\ln (R-\Delta x/2) - \ln \Delta x/2}{R-\Delta x} \simeq \dfrac{\ln N}{R}
	\end{aligned}
\]
Therefore the expectation is asymptotically
larger than $2/R$ for large $N$. 
Moreover, for a large number of counts $M$,
the law of large numbers ensures that the sample average of $1/(x_i-x_A)$ tends to 
its expectation. Therefore, as $M$ increases,
the asymptotic value of $F(a)$ tends to be negative, and 
the external solution of $F(a)=0$ will
be found to the right of the last point of singularity for $F(a)$.
\end{proof}

It is now possible to state a sufficient condition
that applies to uniformly distributed counts in the large--count regime.

\begin{lemma}[Sufficient condition for an acceptable solution]
For large $M$ and $N$ with uniformly distributed counts 
and a non--null count in the last bin,
the external solution of 
$F(a)=0$ is acceptable and it is found to the right of the last singularity.
\label{lemma:sufficient2}
\end{lemma}
\begin{proof} 
For data with non--null counts in the last bin, i.e., $y_N\geq 1$, 
the last singularity 
of $g(a)$ occurs at $a_j=-1/(R-\Delta x_{N}/2)$.
At points of singularity for $g(a)$, $F(a_j)=1+R/2 \cdot a_j$,
according to Equation~\ref{eq:Fa}, and therefore
$F(a_j)>0$. Notice that the point $a_j$ marks the boundary of the region of
acceptability for the solutions of $F(a)=0$, according to 
lemma~\ref{lemma:acceptability}.

The last singularity  for $F(a)$, also a zero of $g(a)$,
 will thus occur at a point $a_s$ 
which is to the \emph{left} of $a_j$, and the continuity of $F(a)$ to the right
of the last singularity ensures that $F(a)$ remains positive between $a_s^+$ and
$a_j$.
Also, lemma~\ref{lemma:average} ensures that the asymptotic value
of $F(a)$ is negative. Therefore there is a zero of $F(a)$ to the right
of $a_j$, and this external zero is acceptable, according to 
lemma~\ref{lemma:acceptability}.

\end{proof}

These asymptotic results apply to a uniform distribution
of counts, which is a very restrictive condition. 
When the distribution of counts is not uniform, 
even large--$M$ datasets may not have
an acceptable model.
It goes beyond the scope of this paper 
to seek additional sufficient conditions for covergence,
given the number of variables at play (in particular, the number
of counts $M$, the number of bins $N$ and their size and location,
and the distribution of counts), and the fact that necessary and
sufficient conditions for convergence have been provided
earlier in this section.
Instead, selected numerical simulations 
are presented to quantify the fraction of Poisson datasets that do not
have a non--negative best--fit linear and 
to illustrate a few representative cases.

For this purpose, $100$ data sets
were simulated for various values of the total number of counts $M$, initially assuming
that the $M$ counts were uniformly distributed among $N=100$ equally spaced bins,
following the same pattern of bins along the $x$ axis as in Figures~\ref{fig:M=2}
and \ref{fig:M=3}. As expected, based on the asymptotic results of this
section, for $M \geq 50$, all datasets have an acceptable model 
(Figure~\ref{fig:acceptance}, red curve). 
Then, the same simulations were repeated for
a distribution of counts that is either linearly 
increasing or decreasing towards larger values of $x$, i.e., with 
samples drawn respectively from the probability distributions functions
\[
h(x)=
\begin{cases}
\dfrac{2(x-x_A)}{R^2}\\
\dfrac{2}{R} -\dfrac{2 (x-x_A)}{R^2}    
\end{cases}
\]
with $x \in [0,R]$.~\footnote{Random samples from these distributions
are readily obtained by simulating the associated normalized linear variables
in $y\in(0,1)$ (with distributions of $2y$ and $2-2y$, respectively for an increasing and 
decreasing distribution). Samples of $x$ are then obtained by rescaling samples of $y$ to the range $R=x_B-x_A$
via a linear transformation with $y=(x-x_A)/R$.
Simulations of the normalized distributions for $y$ are easily 
accomplished with the aid of a uniform variable $u$ in $(0,1)$, which
is commonly available in most software packages. With the aid of the quantile function 
$F^{-1}(p)=y$, where $F$ is the cumulative distribution of $y$ (respectively $F(y)=y^2$ and $F(y)=2y-y^2$
for the two linear models),
the variable $y$ is simulated  as $y=F^{-1}(u)$ (see, e.g., Section 4.8 of \cite{bonamente2017book}).
This means that random samples of the normalized increasing and decreasing distributions
are obtained respectively via $y = \sqrt{u}$ and $y = 1-\sqrt{u}$, where $u$ are samples from
a uniform distribution in $(0,1)$.} 
For these cases, the simulations show that the 
number of acceptable models remains smaller even for large values of $M$.
The right panel of Figure~\ref{fig:acceptance} also illustrates the
fraction of data with a negative $F_{\infty}$. As expected according to  
lemma~\ref{lemma:average}, uniformly distributed
data (red curve) have a negative asymptotic $F_{\infty}$ for large $M$;
moreover, the same applies for data distributed with a negative slope (blue
curve). This is explained according to property~\ref{prop:Finf}, since data points
below the mid--point of the range drive the average of $1/(x_i-x_A)$ 
to values greater than $2/R$, and that in turn causes $F_{\infty}$ to be negative.
For data with a positive slope, even for large $M$ there is a large fraction of data
with a positive asymptotic limit of $F(a)$.
These simulations can be used as examples of large--$M$ data that do not
have an acceptable solution using the linear model of Equation~\ref{eq:yscargle}.

\begin{figure}[!t]
\includegraphics[width=3.2in]{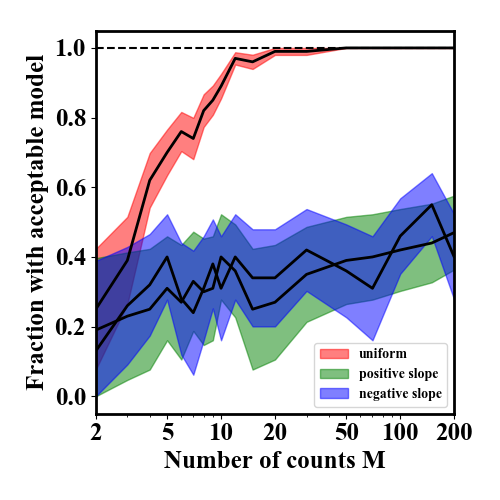}
\includegraphics[width=3.2in]{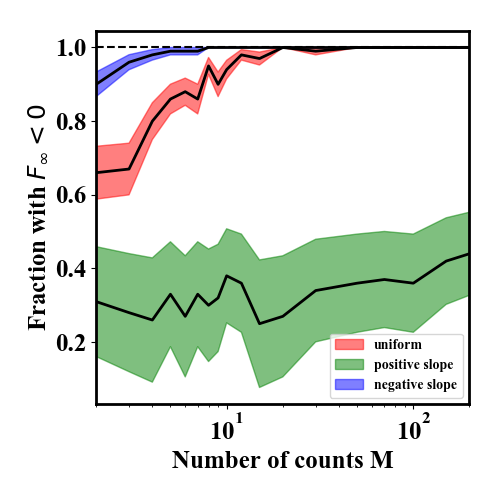}
\caption{(Left:) Fraction of datasets with available best--fit
non--negative linear model, as function of the number of counts $M$.
(Right:) Fraction of datasets with negative asymptotic value $F_{\infty}$.}
\label{fig:acceptance}
\end{figure}

\section{An {\blue extended} linear model with a non--negative solution}
\label{sec:generalized}
The paper has identified
cases where the maximum--likelihood equations do not yield
an acceptable solution for the parameters of the linear model. 
In particular, this is true for all
data with only one count ($M=1$)
and null counts in $N-1$ of the $N$ available bins.
This can simply be viewed as the inability
to constrain two free parameters with just one non--zero data point. 
In such case, it may be sufficient
to model the data with a simple constant model, with a best--fit model
equal to the sample average of the counts in all the bins 
(see, e.g., \cite{bonamente2017book} and \cite{bonamente2020}).
There are also other data sets with $M\geq2$ counts that do not have
an acceptable, non--negative model. One such example was shown in Figure~\ref{fig:M=2}, for a
dataset with $M=2$. Section~\ref{sec:acceptability}
also illustrated data with large $M$ that do not have an acceptable solution (see, e.g.,
Figure~\ref{fig:acceptance}).

Motivated by the need to have a linear model that
is applicable to any situation, this section proposes a simple generalization of the  
linear model of Equation~\ref{eq:yscargle}
that ensures an acceptable maximum--likelihood solution 
using the \cstat\ for any Poisson dataset. 

\begin{definition}[The {\blue extended}  non--negative linear model]
The proposed non--negative linear model is
given by:

(1) the standard linear model of Equation~\ref{eq:yscargle}, when such model
has an acceptable solution; otherwise, 

(2) the model is parameterized as one of the
following three functions:

(A) A one--parameter linear model \emph{pivoted} to zero 
at the initial point $x_A$:

\begin{equation}
f_A(x) = \lambda_A (x-x_A),
\label{eq:ypivotA}
\end{equation}

for which $y_A(x_A)=0$, and with a positive adjustable parameter $\lambda_A\geq 0$.

(B) A one--parameter linear model pivoted to zero 
at the final point $x_B$:

\begin{equation}
f_B(x) = \lambda_B \left( 1 - \dfrac{x-x_A}{R} \right)
\label{eq:ypivotB}
\end{equation}

for which $y_B(x_B) = 0$, with an adjustable parameter $\lambda_B\geq 0$
and therefore a negative slope.

(C) A one--parameter constant model:
\begin{equation}
f_C(x) = \lambda_C.
\label{eq:C}
\end{equation}
\end{definition}

It will be shown that the three models of Equations~\ref{eq:ypivotA}, \ref{eq:ypivotB}
and \ref{eq:C} have simple analytical solutions
for their maximum--likelihood best--fit parameters (respectively
$\lambda_A$, $\lambda_B$ and $\lambda_C$), and therefore
it is always possible to use one of these models as 
an acceptable linear model for any dataset. 

\subsection{Maximum--likelihood solutions for the pivoted
and constant linear models}
\label{sec:MLextension}
For the linear model pivoted at $x_A$,  Equation~\ref{eq:ypivotA}
is used to evaluate the \cstat, Equation~\ref{eq:cstat}.
Assuming that the data covers the range $R$ continuously, as also
assumed for Equation~\ref{eq:summu}, the term
\begin{equation}
\sum_{i=1}^N \mu_i = \int_{x_A}^{x_B} f_A(x) dx = \lambda_A \dfrac{R^2}{2}
\label{eq:summuA}
\end{equation}
 leads to
\begin{equation}
	C_A = \lambda_A R^2 -2M \ln \lambda_{\blue A} + D_A 
\label{eq:CA}
\end{equation}
where
\begin{equation}
 D_A \equiv \left(-2 M 
+ 2 \sum_{i=1}^N y_i \ln y_i - 2 \sum_{i=1}^N y_i \ln \Delta x_i 
-2 \sum_{i=1}^N y_i \ln (x_i - x_A) \right)
\label{eq:DA}
\end{equation}
is a term that is independent of the model, and
therefore plays no role in the minimization of the \cstat.
The best--fit parameter is  given by $\partial C_A/\partial \lambda_A=0$,
leading to the simple analytical solution
\begin{equation}
\lambda_A = \dfrac{2M}{R^2} > 0.
\label{eq:lambdaA}
\end{equation}

For the linear model pivoted at $x_B$, use of Equation~\ref{eq:ypivotB}
into Equation~\ref{eq:cstat} leads to
\begin{equation}
\sum_{i=1}^N \mu_i = \int_{x_A}^{x_B} f_B(x) dx = = \lambda_B \dfrac{R}{2}
\label{eq:summuB}
\end{equation}
and 
\begin{equation} C_B = \lambda_B R -2M \ln \lambda_B + D_B
\label{eq:CB}
\end{equation}
with
\begin{equation} D_B \equiv \left(-2 M 
+ 2 \sum_{i=1}^N y_i \ln y_i - 2 \sum_{i=1}^N y_i \ln \Delta x_i 
-2 \sum_{i=1}^N y_i \ln \left( 1 - \dfrac{x_i - x_A}{R} \right) \right).
\end{equation}
The best--fit parameter is therefore given by
\begin{equation}
\lambda_B = \dfrac{2M}{R} > 0.
\label{eq:lambdaB}
\end{equation}
Finally, the best--fit constant model has a \cstat\ of
\begin{equation}
	C_C = 2 \lambda_C R -2 M \ln {\blue \lambda_c} + D_C
\label{eq:CC}
\end{equation}
with
\begin{equation}
D_C \equiv\left( -2M + 2 \sum_{i=1}^N y_i \ln y_i - 2 \sum_{i=1}^N y_i \ln \Delta x_i \right). 
\end{equation}
This leads to a best--fit parameter
\begin{equation}
\lambda_C = \dfrac{M}{R},
\label{eq:lambdaC}
\end{equation}
which is equivalent to the sample average of the data when multiplied by a uniform $\Delta x$, as 
found in \citep{bonamente2017book}. As already remarked after Equation~\ref{eq:summu},
the equations developed in this section apply to data that cover continuously the range $x_A$ to $x_B$.
Data with gaps in the $x$ variable require a simple modification to these equations that is
presented in Section~\ref{sec:gap}.

\subsection{Use of the {\blue extended} non--negative linear model}

Equation~\ref{eq:yscargle} in combination
with the extensions provided by Equations~\ref{eq:ypivotA}, \ref{eq:ypivotB},
and \ref{eq:C} are to be used according to the following method,
which defines
the solution of the {\blue extended} model.

\begin{definition}[Solution of the {\blue extended} non--negative linear model]
\label{def:generalized}
Solution of the {\blue extended} non--negative linear model is given by:

(1) the solution with the standard linear model of Equation~\ref{eq:yscargle}, if that
solution is acceptable.
As shown in Section~\ref{sec:acceptability} and specifically lemma~\ref{lemma:uniqueness},
this solution is guaranteed to be unique, when it exists.

(2) If a solution with the standard linear model is not available, 
the solution is given by the best--fit model 
that gives the lowest value of the \cstat,
among the three options provided by Equations~\ref{eq:ypivotA}, \ref{eq:ypivotB}
and \ref{eq:C}.
\end{definition}

\begin{lemma}[Existence and uniqueness of solution for the {\blue extended} non--negative
linear model]
There exists one and only one maximum--likelihood 
	solution for the {\blue extended} non--negative linear model
fit to any Poisson--distributed data.
\end{lemma}
\begin{proof} The proof is a direct consequence of the fact that there is at most
one non--negative solution for the model of Equation~\ref{eq:yscargle} (lemma~\ref{lemma:uniqueness}),
	and of definition~\ref{def:generalized} for the solution of the {\blue extended}
	model.
\end{proof}

\begin{remark}[Expanded algorithm for the {\blue extended non--negative} linear model]
\label{remark2}
The algorithm presented in Remark~\ref{remark1} can be extended to
the {\blue non--negative} linear model. When the linear model of Equation~\ref{eq:yscargle}
fails to produce an acceptable solution, the following two additional steps
must be added:
\begin{enumerate}
\setcounter{enumi}{7}
\item Calculate the three additional best--fit linear models (pivoted at A, pivoted at B
and constant) and their \cstat, using the analytical formulas \ref{eq:ypivotA}, \ref{eq:ypivotB}
and \ref{eq:C}.
\item Accept as the best--fit model the one with the lowest \cstat. Notice that if the
original linear model of Equation~\ref{eq:yscargle} is acceptable, its value of the \cstat\
will be lower than that of the other three linear models.
\end{enumerate}
\end{remark}

The use of the {\blue extended} non--negative linear  model
is illustrated in the two following examples.

\begin{example}[Use of the {\blue extended non--negative} model for data with no
acceptable standard linear model]
\begin{figure}[!t]
\includegraphics[width=3.2in]{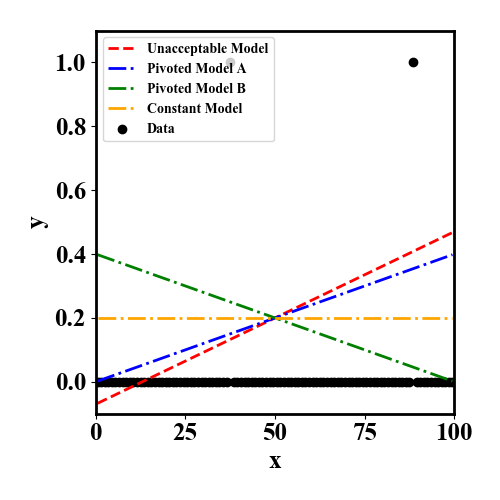}
\includegraphics[width=3.2in]{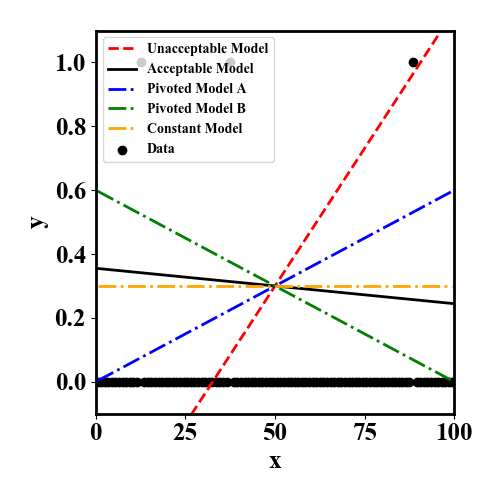}
\caption{Pivoted and constant linear models for the data of
Figure~\ref{fig:M=2} (with $M=2$, left) and Figure~\ref{fig:M=3} 
($M=3$, right).}
\label{fig:extension}
\end{figure}
In the left panel of Figure~\ref{fig:extension} are shown the results for the
same $M=2$ data of Figure~\ref{fig:M=2}, for which a non--negative
linear model according to Equation~\ref{eq:yscargle} could not be found. 
The data can be fit with the pivoted and constant linear models, which yield
best--fit \cstat\ values of $C_A=15.081$, $C_B=18.141$ and $C_C=15.648$.
The values of the \cstat\ indicate that the 
linear model pivoted at $x_A$ is the most accurate representation 
of these data, and should be regarded as the best--fit linear model.
\end{example}

\begin{example}[Use of the {\blue extended non--negative} 
model for data with an acceptable standard linear model]
The right panel of Figure~\ref{fig:extension} shows the
results for the $M=3$ model of Figure~\ref{fig:M=3}, for which
a best--fit non--negative model with the `standard' linear model was
in fact available, for a \cstat\ value of $C=20.996$.
The pivoted and constant linear models yield values of 
$C_A=23.245$, $C_B=22.413$ and $C_C=21.039$, all larger than the value for 
the best--fit standard linear model. This analysis confirms that
the `standard' linear model, when available, is indeed the
most accurate linear representation of the data.
\end{example}

It is in principle possible to devise a linear model different from
those of Equations~\ref{eq:ypivotA}, \ref{eq:ypivotB}
and \ref{eq:C}, that may yield a lower value of the \cstat. There are in fact 
infinitely many such models, e.g., by fixing an arbitrary intercept 
of the $x=x_A$ axis. The choices made by the three simple extensions discussed in
this paper are intended to provide simple alternatives to the full linear model
that have a simple interpretation and likewise simple analytical solutions.

\subsection{Binning and gaps in data}
\label{sec:gap}
The methods of analysis presented in this paper can be applied to data 
with any binning, including data with non--uniform
bin sizes. The bin sizes, however, will have an effect on the best--fit model, as can be seen
by the fact that the function $F(a)$ is a function of $x_j-x_A$, there $x_j$ is the 
center coordinate of the $j$--th bin. 
When Poisson data are collected on an event--by--event basis, the choice of bin size must be
made based on considerations on the methods of collection of the data and the
instruments used for the collection.

In Equations~\ref{eq:summu}, \ref{eq:summuA} and \ref{eq:summuB} it was assumed that the
range of integration of the $x$ variable was continuous, therefore implying that the
data covers the $x_A$ to $x_B$ range without any gaps or missing data. It is possible
to provide a simple generalization to those equations to include gaps in the data.  
This is in fact a situation of practical importance, since certain regions
of the independent variable may be without data for a variety of reasons.
A common situation is the exclusion of portions of the $x$ variable
because of poor calibration of the instrument (e.g.,
the exclusion of a wavelength range because of detector inefficiencies), 
or because an instrument was not operating during certain time intervals.
In these cases, one cannot just assign a value of zero counts to that range
of the independent variable, but rather the intervals must be explicitly removed
from the data, therefore creating gaps in an otherwise continuous variable.

\begin{definition}[Gaps in the data]
A gap in the data is defined as a continuous 
interval of the independent variable
between $x_a$ and $x_b$, of length $R_G=x_b-x_a$, 
that is not covered by any of the bins. 
A Poisson data set may have $g$  non--overlapping gaps 
between $x_{a,j}$ and $x_{b,j}$, $j=1,\dots,g$, with $x_{G,j}=(x_{b,j}+x_{a,j})/2$
the mid--point of each gap and $R_{G,j}=x_{b,j}-x_{a,j}$.
The length of all gaps in the independent variable $x$  is $R_G = \sum R_{G,j}$.
\end{definition}

The following lemmas summarize the changes that need to be made to analyze data that contain
gaps in the independent variable

\begin{lemma}[Modifications to the \cstat\ and to the functions $F(a)$ and $\lambda(a)$ for gaps in the data]
When the data have gaps, the \cstat\ becomes
\label{lemma:Gap1}
\begin{multline}
C = 2\lambda R\left(1+\frac{aR}{2}\right) -2 \lambda \sum_{j=1}^g R_{G,j} (1+a(x_{G,j}-x_A))  \\
-2 M \ln \lambda
        -2\sum_{i=1}^N  y_i \ln (1+a(x_i-x_A)) + D.
\label{eq:cstatscargleGap}
\end{multline}
Moreover, the function whose zero provides the best--fit value of  $a$ becomes
\begin{equation}
F(a) =  1 + a\dfrac{R_{\blue m}}{2} - \dfrac{MR_{\blue m}}{2 g(a)}
\label{eq:FaGap}
\end{equation}
	where {\blue $R$ is replaced by a modified $R_m$ given by}
\begin{equation}
\begin{cases}
	R_{\blue m}\equiv \dfrac{R^2-2 S_G}{R-R_G} \\ 
S_G\equiv \sum_{j=1}^g R_{G,j} \left(x_{G,j}-x_A\right),
\end{cases}
\label{eq:Rprime}
\end{equation}
and the best--fit solution for the parameter $\lambda$ is 
\begin{equation}
\lambda(a) = \frac{M}{R\left(1+a \dfrac{R}{2}\right) - (R_G + a S_G)}.
\label{eq:lambdaGap}
\end{equation}
\end{lemma}
\begin{proof}
The modification to the \cstat\ to account for the presence of
gaps is provided by changing equation~\ref{eq:summu} to
\begin{dmath}
\sum_{i=1}^N \mu_i = \int_{x_A}^{x_B} f(x) dx - \sum_{j=1}^g \int_{x_{a,j}}^{x_{b,j}} f(x) dx = 
\lambda R \left( 1 + a \dfrac{R}{2} \right) - \lambda \sum_{j=1}^g R_{G,j} ( 1+ a(x_{G,j}-x_A)) 
\label{eq:summuGap}
\end{dmath}
The use of Equation~\ref{eq:summuGap} in place of Equation~\ref{eq:summu} leads to
the \cstat\ of Equation~\ref{eq:cstatscargleGap} in place of the original equation~\ref{eq:cstatscargle}.
Taking the derivatives of $C$ with respect to $a$ and $\lambda$ and setting them to zero 
leads to
\[ \frac{\partial C}{\partial \lambda} = 2 R \left( 1 + a \dfrac{R}{2}\right) 
-2 \sum_{j=1}^g R_{G,j} ( 1+ a(x_{G,j}-x_A))  -2\frac{M}{\lambda} = 0, 
\]
and
\[ \frac{\partial C}{\partial a} = \lambda R^2 - 2\lambda \sum_{j=1}^g R_{G,j}  (x_{G,j}-x_A) 
-2 \sum_{i=1}^N y_i \dfrac{(x_i-x_A)}{1+a(x_i-x_A)} =0.
\]
Notice that 
\[ \sum_{j=1}^g R_{G,j} ( 1+ a(x_{G,j}-x_A)) = R_G + a \sum_{j=1}^g R_{G,j} (x_{G,j}-x_A) \]
where $R_G$ is the combined length of all (non--overlapping) gaps.
Defining
\begin{equation}S_G \equiv  \sum_{j=1}^g R_{G,j} (x_{G,j}-x_A) 
\label{eq:SG}
\end{equation}
leads to 
\[   R \left( 1 + a \dfrac{R}{2}\right) - R_G - a S_G - \dfrac{M}{\lambda}=0, \]
	{\blue thus proving Equation~\ref{eq:lambdaGap}}, and 
\[  \lambda R^2 - 2\lambda S_G - 2 g(a) =0 \]
where $g(a)$ is the usual function as defined in Equation~\ref{eq:ga}.
Simple algebraic modifications and elimination of $\lambda$ lead to  
\[
1 + a\dfrac{R_{\blue m}}{2} - \dfrac{MR_{\blue m}}{2 g(a)} = 0
\]
where
\[
R_{\blue m}\equiv \dfrac{R^2-2 S_G}{R-R_G},
\]
thus proving Equation~\ref{eq:FaGap}. 
\end{proof}
Lemma~\ref{lemma:Gap1} shows that, when there are gaps in the independent variable,
the method of analysis to find a solution for $a$ and $\lambda$ proceeds in the same way
as when there are no gaps,
provided the function $F(a)$ uses the $R_{\blue m}$ parameter in place of $R$. Once the 
best--fit value of $a$ is found, $\lambda$ can be calculated analytically by making a
change in the denominator of the function $\lambda(a)$ to account for the gap $R_G$, according
to equation~\ref{eq:lambdaGap}.

\begin{lemma}[Modifications to the \cstat\ and to the best--fit parameters of the pivoted and constant models
for gaps in the data]
\label{lemma:Gap2}
When the data have gaps, the \cstat\ for the pivoted and constant models become
\begin{equation}
\begin{cases}
C_A = \lambda_A R^2 - \lambda_A S_A^2 - 2 M \log \lambda_A + D_A\\
C_B = \lambda_B R -2\lambda_B S_B - 2 M \ln \lambda_B + D_B \\
C_C = 2 \lambda_C (R - R_G) - 2 M \ln \lambda_C + D_C
\end{cases}
\label{eq:Gap2C}
\end{equation}
with
\begin{equation}
\begin{cases}
S_A^2 \equiv \sum_{j=1}^g x_{b,j}^2 - x_{a,j}^2\\
S_B \equiv \sum_{j=1}^g \dfrac{R_{G,j}}{R} (x_B - x_{G,j}).
\end{cases}
\end{equation}
The best--fit model parameters become
\begin{equation}
\begin{cases}
\lambda_A = \dfrac{2M}{R^2 - 2 S_G} \\
\lambda_B = \dfrac{2M}{R - 2 S_B}\\
\lambda_C = \dfrac{M}{R-R_G}.
\end{cases}
\label{eq:Gap2BestFit}
\end{equation}
\end{lemma}
\begin{proof}
For the model pivoted at $A$, equation~\ref{eq:summuA} is modified by the presence of gaps as
\begin{dmath}
\sum_{i=1}^N \mu_i = \int_{x_A}^{x_B} f_A(x) dx - \sum_{j=1}^g \int_{x_{a,j}}^{x_{b,j}} f_A(x) dx= 
\lambda_A \left( \dfrac{x_B^2-x_A^2}{2} - R x_A \right) -\lambda_A \left( \sum_{j=1}^g \dfrac{x_{b,j}^2-x_{a,j}^2}{2}
- \sum_{j=1}^g x_A (x_{b,j}-x_{a,j})\right).
\label{eq:summuAGap}
\end{dmath}
Defining
\begin{equation}
S_A^2 \equiv \sum_{j=1}^g (x_{b,j}^2-x_{a,j}^2)
\label{eq:SA}
\end{equation}
and noticing that
\[ \sum_{j=1}^g x_A (x_{b,j}-x_{a,j}) = x_A R_G \]
leads to 
\[
C_A = \lambda_A (x_B^2-x_A^2) - \lambda_A S_A -2 \lambda_A x_A (R-R_G) -2 M \ln \lambda_A +D_A.\]
Since $S_A^2 - 2 x_A R_G = 2 S_G$ and $(x_B^2-x_A^2) - x_A R = R^2$, 
it follows that 
\[ C_A = \lambda_A (R^2 -2 S_G) -2 M \ln \lambda_A +D_A.\]
Then, taking a derivative of $C_A$ with respect to $\lambda_A$ and setting
it to zero completes the proof for the model pivoted at A.

For the model pivoted at $B$, 
\begin{dmath}
\sum_{i=1}^N \mu_i = \int_{x_A}^{x_B} f_B(x) dx - \sum_{j=1}^g \int_{x_{a,j}}^{x_{b,j}} f_B(x) dx = 
\lambda_B \left(R+x_A- \dfrac{x_B^2-x_A^2}{2R}\right) 
- \lambda_B \sum_{j=1}^g \left(1+\dfrac{x_A}{R}\right)R_{G,j}
-\left(\dfrac{x_{b,j}^2-x_{a,j}^2}{2 R} \right) 
= \lambda_B \left( \dfrac{R}{2} - S_B \right)
\label{eq:summuBGap}
\end{dmath}
where 
\[ S_B \equiv \sum_{j=1}^g R_{G,j}\left(1 - \dfrac{x_{G,j}-x_A}{R} \right) = \sum_{j=1}^g \dfrac{R_{G,j}}{R}(x_B-x_{G,j}).
\]
From this, the equations for $C_B$ and $\lambda_B$ follow after a few simple algebra steps.

The results for the constant model follow immediately from the 
constancy of the function $f_C(x)=\lambda_C$.
\end{proof}
Lemma~\ref{lemma:Gap2} shows that the pivoted and constant models retain a simple
analytical solution even in the presence of gaps in the data. 
An application of the fit to Poisson data with non--uniform bin sizes and with a gap in the
data is provided in the following example.

\begin{example}[Data with non--uniform bin sizes and a gap in the data]
\begin{figure}[!t]
\includegraphics[width=3in]{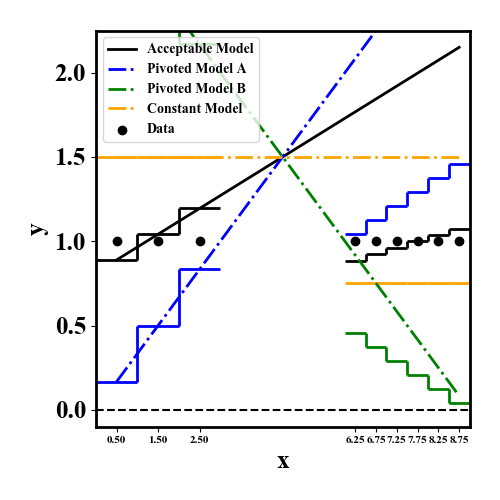}
\includegraphics[width=3in]{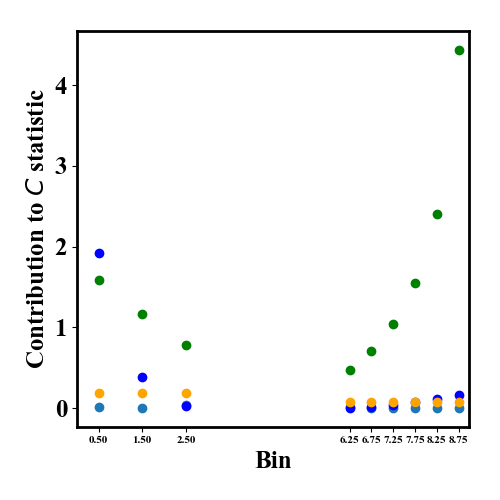}
\caption{Best--fit linear models for data with non--uniform bins and with a gap in the data.
The dot--dashed curve are the density functions and the solid step--wise curves are
the models $y(x_i)$ for the integer data. }
\label{fig:gap}
\end{figure}

The data chosen for this example span a range of the independent variable
between $x_A=0$ and $x_B=9$, with a gap between $x_a=3$ and $x_b=6$.
All nine measurements have a value of $y_i=1$, with bin sizes of $\Delta x_i=1$ for the
first three data points, and $\Delta x_i=\nicefrac{1}{2}$ for the other six data points,
as shown in Figure~\ref{fig:gap}. The data have an acceptable solution for the
standard linear model (in black) with $a=0.188$, $\lambda=0.812$, 
for a best--fit statistic of $C_{\mathrm cmin}=0.078$.
Given the non--uniform bin sizes, the best--fit density function $f(x)$ (black continuous line, in units
of counts--per--bin--size) differs from the best--fit model $y(x_i)$ (black step--wise curve,
in units of counts or counts--per--bin).
The constant model (yellow) has a best--fit parameter
of $\lambda_C=1.5$, according to equation~\ref{eq:Gap2BestFit} with $M=9$, $R=9$ and $R_G=3$,
with $C_C=1.019$,
the linear model pivoted at $A$ has $\lambda_A = 0.333$  and $C_A= 2.735$,
and the linear model pivoted at $B$ has $\lambda_B = 3$ , and $C_B= 14.177$.
\end{example}

In summary, there are no significant additional complication for the analysis of data that
contain a number of gaps or missing data. The following algorithm summarizes the changes 
required to analyze data with gaps.

\begin{remark}[Algorithm to implement changes in the analysis when gaps in the data are present]
\label{remark3}
This algorithm details the additions and modifications required for algorithms~\ref{remark1}
and \ref{remark2}
when there are data gaps present, following the same enumeration.
\begin{enumerate}
\setcounter{enumi}{-1}
\item (Additional step) Calculate the location $x_{a,j}$, $x_{b,j}$ and range $R_{G,j}$ of each gap,
the total gap length $R_G$,  $R_{\blue m}$ and $S_G$ according to equation~\ref{eq:Rprime}.
\setcounter{enumi}{3}
\item Hereafter replace $R$ with $R_{\blue m}$ in the definition of $F(a)$.
\setcounter{enumi}{6}
\item Use equation~\ref{eq:lambdaGap} instead of \ref{eq:lambda} to calculate the
value $\lambda(a)$ corresponding to an acceptable solution $a$.
\setcounter{enumi}{7}
\item For the calculation of the $C$ statistics and best--fit parameters 
of the constant and pivoted models,
use respectively Equations~\ref{eq:Gap2C} (instead of Equations~\ref{eq:CA}, \ref{eq:CB} 
and \ref{eq:CC}) and Equations~\ref{eq:Gap2BestFit} (instead of Equations~\ref{eq:lambdaA},
\ref{eq:lambdaB} and \ref{eq:lambdaC}).
\end{enumerate}

\end{remark}

\subsection{A note on the distribution of the \cstat}

It is well known that, in the large--count limit, the \cmin\ statistic
 -- i.e., the \cstat\ evaluated for the best--fit linear model -- 
is expected to be 
distributed like a $\chi^2$ distribution with $N-2$ degrees of freedom, where
$N$ is the number of bins and $2$ is the
number of adjustable free parameters of the linear model
(e.g., \cite{cash1979} and \cite{bonamente2017book}). 
Moreover, properties  of the \cstat\ for a fixed model with no free parameters
is also known accurately for any value of the parent Poisson 
mean \citep{kaastra2017, bonamente2020}.
What remains to be analyzed in further detail is the effect of free parameters
on the distribution of \cmin, in the low--count regime. 
The purpose of this paper is to present a method to evaluate the best--fit parameters
of the linear model, precisely with the intent to further study the distribution
of \cmin\ via numerical simulations that rely on this method of analysis.

For a significant number of data sets, and especially for data with a small number of counts,
 the only non--negative linear model is one of the
three extensions -- all of them with just one adjustable parameter, instead of
two of the traditional linear model. This requirement that the model
be non--negative was introduced by the use of the Poisson distribution, and 
did not enter the discussion of Gaussian--distributed datasets 
that can be fit with the $\chi^2$ distribution. It is likely that
such new requirement will result in differences between the
distributions of $\chi^2_{min}$ and \cmin\ for the linear model in the low--count regime,
with implications for hypothesis testing and confidence intervals
on the best--fit parameters.
The distribution of the \cmin\ for the linear model
 will be presented in a separate paper.

\section{Discussion and conclusions}
\label{sec:discussion}
This paper has presented a new semi--analytical method to find the best--fit 
parameters of a linear model for the fit to integer--valued counting data, using
the Poisson--based \cstat. The method consists first of finding a solution for the non-linear equation
$F(a)=0$, where $a$ is one of the two parameters of the model. 
 The other
parameter $\lambda$ is then calculated analytically via a simple analytical 
function $\lambda=\lambda(a)$.
The two parameters $a$ and $\lambda$ must be such that the linear model is non--negative in
each bin, in order to ensure the applicability of the Poisson distribution. The analysis presented
in this paper shows that such requirement leads, in fact, to the \emph{uniqueness} of the
best--fit model, when such solution is available. 
This is clearly a very desirable property of the method, and a necessary condition
for the use of this method to analyze Poisson--distributed  data.

This paper has identified cases where low--count Poisson data do not have a suitable non--negative
best--fit linear model according to the standard
parameterization of Equation~\ref{eq:yscargle}. 
For this reason, an {\blue extended}  linear model was proposed
that guarantees a unique non--negative solution for any Poisson data set.
This is accomplished by pivoting the linear model
to either end of the range of the independent variable or by using a simple constant linear model,
when the traditional linear model leads to an unsuitable solution. Thanks to simple analytical
solutions for the best--fit parameter of these extensions, the use of the 
{\blue extended non--negative} linear model
remains straightforward.

The availability of a simple method to identify the best--fit parameters of a linear model
for Poisson data of any number of counts makes it possible to further our understanding of the \cstat.
In particular, it is now possible to study the 
distribution of the \cmin\ statistic
for one of the most commonly used models with adjustable parameters, i.e., the linear model,
especially in the low--count regime where its distribution is not known exactly.


\begin{thebibliography}{10}
\providecommand{\MR}{\relax\unskip\space MR }
\providecommand{\url}[1]{\normalfont{#1}}
\providecommand{\urlprefix}{Available at }

\bibitem{cousins1984}
S. {Baker} and R.D. {Cousins}, \emph{{Clarification of the use of CHI-square
  and likelihood functions in fits to histograms}}, Nuclear Instruments and
  Methods in Physics Research 221 (1984), pp. 437--442.

\bibitem{bevington2003}
P.R. {Bevington} and D.K. {Robinson}, \emph{Data reduction and error analysis
  for the physical sciences}, McGraw Hill, Third Edition, 2003.

\bibitem{bonamente2017book}
M. {Bonamente}, \emph{Statistics and Analysis of Scientific Data}, Springer,
  Graduate Texts in Physics, Second Edition, 2017.

\bibitem{bonamente2020}
M. Bonamente, \emph{Distribution of the c statistic with applications to the
  sample mean of poisson data}, Journal of Applied Statistics 0 (2019), pp.
  1--22. \urlprefix\url{https://doi.org/10.1080/02664763.2019.1704703}.

\bibitem{bonamente2019}
M. Bonamente, \emph{Probability models of chance fluctuations in spectra of
  astronomical sources with applications to x-ray absorption lines}, Journal of
  Applied Statistics 46 (2019), pp. 1129--1154.
  \urlprefix\url{https://doi.org/10.1080/02664763.2018.1531976}.

\bibitem{bonat2018}
W.H. Bonat, B. Jorgensen, C.C. Kohonendji, J. Hinde, and C. Demetrio,
  \emph{Extended poisson-tweedie: Properties and regression models for count
  data}, Statistical Modeling 18(1) (2018), pp. 24--49.

\bibitem{cash1976}
W. {Cash}, \emph{{Generation of Confidence Intervals for Model Parameters in
  X-ray Astronomy}}, \aap 52 (1976), p. 307.

\bibitem{cash1979}
W. {Cash}, \emph{Parameter estimation in astronomy through application of the
  likelihood ratio}, \apj 228 (1979), p. 939.
  \urlprefix\url{http://adsabs.harvard.edu/cgi-bin/nph-bib_query?bibcode=1979ApJ...228..939C&db_key=AST}.

\bibitem{dobson2018}
A. {Dobson} and A. {Barnett}, \emph{An Introduction to Generalized Linear
  Models}, CRC Press, Fourth Edition, 2018.

\bibitem{elsayyad1973}
G.M. El-Sayyad, \emph{Bayesian and classical analysis of poisson regression},
  Journal of the Royal Statistical Society. Series B (Methodological) 35
  (1973), pp. 445--451. \urlprefix\url{http://www.jstor.org/stable/2985109}.

\bibitem{haselimashhadi2018}
H. Haselimashhadi, V. Vinciotti, and K. Yu, \emph{A novel bayesian regression
  model for counts with an application to health data}, Journal of Applied
  Statistics 45 (2018), pp. 1085--1105.
  \urlprefix\url{https://doi.org/10.1080/02664763.2017.1342782}.

\bibitem{humphrey2009}
P.J. {Humphrey}, W. {Liu}, and D.A. {Buote}, \emph{{{$\chi$}$^{2}$ and
  Poissonian Data: Biases Even in the High-Count Regime and How to Avoid
  Them}}, \apj 693 (2009), pp. 822--829.

\bibitem{kaastra2017}
J.S. {Kaastra}, \emph{{On the use of C-stat in testing models for X-ray
  spectra}}, \aap 605 (2017), A51.

\bibitem{mccullagh1989}
P. {McCullagh} and J. {Nelder}, \emph{Generalized Linear Models}, Chapman \&
  Hall/CRC, Second Edition, 1989.

\bibitem{mock2011}
D.M. Mock, M. N.I., Z. S., and  et  al., \emph{Red blood cell (rbc) survival
  determined in humans using rbcs labeled at multiple biotin densities.},
  Transfusion 51(5) (2011), pp. 1047--1057.

\bibitem{nelder1972}
J.A. Nelder and R.W.M. Wedderburn, \emph{Generalized linear models}, Journal of
  the Royal Statistical Society. Series A (General) 135 (1972), pp. 370--384.
  \urlprefix\url{http://www.jstor.org/stable/2344614}.

\bibitem{scargle2013}
J.D. {Scargle}, J.P. {Norris}, B. {Jackson}, and J. {Chiang}, \emph{{Studies in
  Astronomical Time Series Analysis. VI. Bayesian Block Representations}}, \apj
  764 (2013), 167.

\bibitem{sellers2010}
K.F. Sellers and G. Shmueli, \emph{A flexible regression model for count data},
  The Annals of Applied Statistics 4 (2010), pp. 943--961.
  \urlprefix\url{http://www.jstor.org/stable/29765537}.

\bibitem{shmueli2005}
G. Shmueli, T.P. Minka, J.B. Kadane, S. Borle, and P. Boatwright, \emph{A
  useful distribution for fitting discrete data: Revival of the
  conway-maxwell-poisson distribution}, Journal of the Royal Statistical
  Society. Series C (Applied Statistics) 54 (2005), pp. 127--142.
  \urlprefix\url{http://www.jstor.org/stable/3592603}.

\bibitem{valenti2016}
S. {Valenti}, D.A. {Howell}, M.D. {Stritzinger}, M.L. {Graham}, G.
  {Hosseinzadeh}, I. {Arcavi}, L. {Bildsten}, A. {Jerkstrand}, C. {McCully}, A.
  {Pastorello}, A.L. {Piro}, D. {Sand}, S.J. {Smartt}, G. {Terreran}, C.
  {Baltay}, S. {Benetti}, P. {Brown}, A.V. {Filippenko}, M. {Fraser}, D.
  {Rabinowitz}, M. {Sullivan}, and F. {Yuan}, \emph{{The diversity of Type II
  supernova versus the similarity in their progenitors}}, \mnras 459 (2016),
  pp. 3939--3962.

\bibitem{yee2015}
T. Yee, \emph{Vector Generalized Linear and Additive Models}, Springer, 2015.

\bibitem{yee1996}
T.W. Yee and C.J. Wild, \emph{Vector generalized additive models}, Journal of
  the Royal Statistical Society. Series B (Methodological) 58 (1996), pp.
  481--493. \urlprefix\url{http://www.jstor.org/stable/2345888}.

\end{thebibliography}
\if\JAS1
{
\bibliographystyle{tfs}
} \fi

\if\JASA1
{
\bibliographystyle{tfs}
} \fi

\end{document}